\pdfoutput=1

\documentclass{article}

\usepackage[a4paper,inner=80pt,outer=80pt,textheight=630pt,centering]{geometry}

\usepackage{color}
\usepackage{ownstyle}

\begin{document}

\title{From signatures to monads in \UniMath}
\author{Benedikt Ahrens\thanks{This material is based upon work supported by the
National Science Foundation under agreement Nos. DMS-1128155 and CMU 1150129-338510.
Any opinions, findings and conclusions or recommendations expressed in this material
are those of the author(s) and do not necessarily reflect the views of the National
Science Foundation.}~\thanks{This work has partly been funded by the CoqHoTT ERC Grant 637339.}
\and Ralph Matthes \and Anders Mörtberg\footnotemark[1]}


\maketitle

\begin{abstract}
The term \UniMath refers both to a formal
system for mathematics, as well as a computer-checked library of
mathematics formalized in that system. The \UniMath system is a core
dependent type theory, augmented by the univalence axiom. The system
is kept as small as possible in order to ease verification of it---in
particular, general inductive types are not part of the system.

In this work, we partially remedy the lack of inductive types 
by constructing some datatypes
and their associated induction principles from other type
constructors. This involves a formalization of a category-theoretic
result on the construction of initial algebras, as well as a mechanism
to conveniently use the datatypes obtained. We also connect this
construction to a previous formalization of substitution for languages
with variable binding. Altogether, we construct a framework that
allows us to concisely specify, via a simple notion of binding
signature, a language with variable binding. From such a specification
we obtain the datatype of terms of that language, equipped with a
certified monadic substitution operation and a suitable recursion
scheme. Using this we formalize the untyped lambda calculus and the
raw syntax of Martin-Löf type theory.
\end{abstract}

\tableofcontents 

\section{Introduction}

The \UniMath\footnote{The \UniMath library can be found at:
  \url{https://github.com/UniMath/UniMath}\\A summary file related to
  this paper can be found at:
  \url{https://github.com/UniMath/UniMath/blob/master/UniMath/SubstitutionSystems/FromBindingSigsToMonads_Summary.v}}
language is meant to be a core dependent type theory, making use of
as few type constructors as possible.  The goal of this restriction to
a minimal \enquote{practical} type theory is to make a formal proof of
(equi-)consistency of the theory feasible.  In practice, the \UniMath
language is (currently) a subset of the language implemented by the
proof assistant \Coq.
Importantly, the \UniMath language does not include a primitive for postulating arbitrary inductive types.
Concretely, this means that the use of the \Coq \texttt{Inductive} vernacular is not
part of the subset that constitutes the \UniMath language.
The purpose of avoiding the  \texttt{Inductive} vernacular is
to ease the semantic analysis of \UniMath, that is, the construction of models of the \UniMath language.
Another benefit of keeping the language as small as possible is that
it will be easier to one day port the library to a potential
proof assistant specifically designed for univalent mathematics.

In the present work, we partially remedy the lack of general inductive
types in \UniMath
by constructing datatypes as initial algebras.
We provide a
suitable induction principle for the types we construct, analogous to
the induction principle the \texttt{Inductive} scheme would generate
for us.  This way we can construct standard datatypes, for instance
the type of lists over a fixed type, with reasonable computational
behavior as explained in Section~\ref{sec:conclusion}.
In what follows we refer to types defined using \Coq's
\texttt{Inductive} scheme as ``inductive types'' and the types we
construct as ``initial algebras'' or ``datatypes''.

Intuitively, datatypes are types of tree-shaped data, and \emph{inductive} datatypes limit them to wellfounded
  trees; here we exemplify two use cases:
\begin{itemize}
  \item Structured collections of homogeneous data, e.\,g., lists of
    elements of a fixed type:
  \begin{verbatim}
    Inductive list (X : Type) :=
      | nil : list X
      | cons : X -> list X -> list X.
  \end{verbatim}
  \vspace{-10pt} 

    There are also many kinds of branching data structures for
    organizing homogeneous data.

  \item Representations of mathematically interesting objects, e.\,g.,
    natural numbers and lambda terms (see
    Example~\ref{ex:sig_strength_lc} for a categorical presentation)
    where the type parameter represents the names
    of the variables that may occur free in them:
    \begin{verbatim}
    Inductive LC (X : Type) :=          
      | Var : X -> LC X
      | App : LC X * LC X -> LC X
      | Abs : LC (option X) -> LC X    
    \end{verbatim}
    \vspace{-10pt}
    Here \verb+option X+ is \verb+X+ together with one extra element. This is an example of a ``nested datatype'' (see Section~\ref{sec:inductive_families}).
\end{itemize}

There are two ways to characterize (or specify) inductive datatypes:
either \emph{externally}, via inference rules, or \emph{internally},
via a universal property.
The relationship between the two ways was studied in
\cite{DBLP:conf/lics/AwodeyGS12}. 
There, the authors do not ask
whether (some) inductive types are derivable in univalent
mathematics.
Instead, they 
start with a basic type theory with the axiom of function extensionality,
and 
present two extensions of that type theory
by axioms postulating inductive types, in two different ways: first by
axioms mimicking the inference rules, that is, by an internal variant
of an external postulate, and second by axioms postulating existence of initial
algebras for polynomial functors.  The authors then show that those
extensions are (logically) equivalent.
In the present work, we are
interested in an internal characterization of
datatypes, as initial objects, and we construct suitable initial algebras.

An inductive datatype has to come with a \emph{recursion principle} (a
calculational form of the universal property) which ought to be
mechanically derived together with the datatype itself.  Doing this by hand on
a case-by-case basis means doing similar tasks many times.
For the research program that tries to avoid this ``boiler plate''
of multiple instances of the same higher-level principles, the name
\enquote{datatype-generic programming} has been coined by Roland
Backhouse and Jeremy Gibbons---nicely indicating in what sense
genericity is aimed at.

In this work we focus on a particular class of datatypes that represent languages with
variable binding. Those datatypes are families of types that are indexed over the type
of free variables allowed to occur in the expressions of the language. Variable binding
modifies the indexing type by adding extra free variables in the scope of the binder, as
seen in the motivating code example \verb+LC+ of representations of lambda terms above.
 
Still within the target area of datatype-generic programming (and
reasoning), but more specifically, the datatypes we focus on in the
present work are canonically equipped with a \emph{substitution}
operation---itself defined via a variant of the recursion principle
associated to the datatypes (recursion in Mendler-style
\cite{DBLP:journals/apal/Mendler91}).  This substitution satisfies the
laws of the well-known mathematical structure of a \emph{monad}---an
observation originating in
\cite{BellegardeHook,DBLP:journals/jfp/BirdP99,alt_reus}.
In this work, we not only construct the datatypes themselves, but
also provide a monadic structure---both the operations and the
laws---on those datatypes.

The datatypes representing languages with binders are specified
via a notion of \emph{signature}. A signature abstractly describes the
shape of the trees by specifying
\begin{itemize}
  \item the type of nodes and
  \item the ``number'' of subtrees of a node.
\end{itemize}
In the present work, we consider two notions of signatures,
and relate
them by constructing a function from one type of signatures to the
other.  One notion is that of a \emph{binding signature}
(cf.~Definition~\ref{def:binding_sig}), a
simple notion of
signature for which we know how to construct their associated
datatypes.  The other notion is that of a \emph{signature with
  strength} (cf.~Definition~\ref{def:sig_w_strength}), introduced in
\cite{DBLP:journals/tcs/MatthesU04}.  The latter is a more general
notion of signature which comes with information on how to perform
\emph{substitution} on the associated language (or, more generally, on
any \enquote{model} of the signature---even including coinductive
interpretations in form of languages with non-wellfounded legal parse trees,
that, however, are not studied in the present work).

\paragraph*{Outline of the paper.}
~The present work is built on top of existing work.
Here, we list previous work as well as
work done for the present article, in order to give a coherent
picture:

\begin{enumerate}
  \item In Section~\ref{sec:signature}, we construct a signature with
    strength from a binding signature.  This involves constructing an
    endofunctor on the category of endofunctors on a base category $\C$
    from a family of lists of natural numbers, as well as a strength (a
    natural transformation with extra properties) between suitable
    functors. \label{item:signatures}
      
  \item Instantiating the base category $\C$ of the previous section to $\set$, 
    we construct the data type, as an initial algebra of the
    endofunctor on endo\-functors on the category of sets that is
    specified by a binding signature, using just the type constructors
    available in the \UniMath language.  In particular, we do not use
    the \Coq vernacular \verb!Inductive!.
    This work is reported on in Section~\ref{sec:inductive_types}. \label{item:init_alg}
    
  \item

    In previous work
    \cite{DBLP:journals/tcs/MatthesU04,DBLP:journals/corr/AhrensM16},
    a model (``substitution system'') of a signature with strength was
    constructed on a hypothetical initial algebra.  This construction
    was carried out over an arbitrary base category, which, by
    hypothesis, is sufficiently well-behaved. In particular, right
    Kan extensions were required to exist. In the present work, we 
    base the needed scheme of generalized iteration in Mendler-style
    on another theorem in~\cite{DBLP:journals/fac/BirdP99} that is
    based on cocontinuity assumptions instead of the existence of right
    Kan extensions.
    We apply this modified construction to the data type
    constructed in item \ref{item:init_alg}, 
    where the base category is the category of sets.\label{item:hss} 
    We hence have to provide the prerequisites for that general
    construction, in particular, we show that precomposition with
    a functor preserves colimits of any kind (while only preservation
    of initial objects and colimits of chains is required for the iteration scheme).
    This work is reported on in Section~\ref{sec:hss}.
  
    \item In previous work \cite{DBLP:journals/tcs/MatthesU04,DBLP:journals/corr/AhrensM16}, a monad was constructed from any substitution system 
          over an arbitrary base category---thus showing that the substitution constructed in \ref{item:hss} satisfies 
          widely recognized minimum requirements on substitution.
          The modified construction of the present work can be applied to our more specific situation without any further conceptual work, see Section~\ref{sec:monad}.
          \label{item:monad}
      
   \end{enumerate}

The construction of (indexed) datatypes described in item \ref{item:init_alg} 
certainly constitutes the 
bulk of the present work,
but connecting this construction with the previous work mentioned above also required some effort.
The construction is done by combining two results:
  
\begin{itemize}
  \item a classical category-theoretic result saying that an initial
    algebra of an $\omega$-co\-con\-tinuous functor can be constructed
    from a colimit of a certain chain (i.\,e., a countably
    infinite linear diagram)
    \cite{Adamek74};
  \item the constructibility of colimits in the category of sets
    (a.k.a.~discrete types) in \UniMath as a consequence of the
    constructibility of set quotients.
\end{itemize}

The construction of set-level quotients was done by Voevodsky
\cite{voevodsky_experimental}.  It is a prime example of the new
possibilities that the univalence axiom and its consequences provide
for the formalization of (set level) mathematics compared to the type
theories implemented by \Coq or \Agda without the univalence axiom.

On the way to our results, we also deepened the degree of categorical
analysis, e.\,g., we organized the signatures with strength into a
category, constructed certain limits and colimits in that category, and identified pointed
distributive laws as a means to construct signatures with strength.

The results presented in this article are \emph{not
  surprising}---it is our hope, however, that their formalization will
be \emph{useful} and that its underlying ideas extend to richer
notions of datatypes and type families. 
One envisioned use of the library formalized in the present work
is outlined in Section~\ref{sec:c-systems}.

\subsection{About \UniMath}

The \UniMath language is a variation of intensional Martin-Löf type
theory~\cite{MLTT}. 
It has dependent function types (also called
$\Pi$-types), dependent pair types (also called $\Sigma$-types),
identity types and coproduct types.

There are also a few base types: the type of natural numbers, the
empty type, the unit type, the type of booleans.  Furthermore, we
assume that all the types are elements of a universe $\UU$---for sake
of simplicity, and while waiting for a satisfying universe mechanism
(that supports resizing rules besides universe polymorphism), we even
assume the inconsistent typing rule $\UU:\UU$. This means that the
\Coq system does not provide us with a validation of our usage of
universes, although we informally claim that we do not exploit that
rule in inconsistent ways (to be confirmed in future implementations).
We denote by $A \simeq B$ the type of equivalences between types
$A$ and $B$. For details we refer to \cite[Chapter 2.4]{hottbook}.

An important part of \UniMath is the \fat{univalence axiom}. 
This axiom characterizes the identity type on the universe $\type$.
It asserts that the type of identities between types is 
equivalent to the type of equivalences between those types.
More precisely, it asserts that the map from identities between types
to equivalences between types that is specified by sending the 
reflexivity to the identity equivalence is an equivalence.
In the present work, we crucially use some consequences of the 
univalence axiom that are not provable in pure Martin-Löf type theory.
Details are described in Section~\ref{sec:conclusion}.
  
Note that the general scheme to define strictly positive inductive
types and families in \Coq, the vernacular \texttt{Inductive}, is
\fat{not} part of \UniMath.  Indeed, while the types above are, for
technical reasons, implemented in \UniMath using the
\texttt{Inductive} vernacular, its use is not permitted outside a
\enquote{preamble} that introduces those types.  In this way we
simulate a theory in which the above types are primitive rather than
an instance of a general type definition mechanism.  It is the purpose
of the present work to construct some of the inductive types that
could otherwise be defined using the \texttt{Inductive} scheme.\footnote{
This is similar in spirit to the datatype mechanism of the
\Isabelle proof assistant where the datatypes are constructed inside a
core theory; thus the recursion and induction principles do not form
part of the \enquote{trusted code base} of \Isabelle while they do
constitute a part of the \Coq kernel. We go beyond the justification in \Isabelle
in having the base category as parameter.}
Consequently, the experimental
\emph{Higher Inductive Types (HITs)}~\cite{hottbook} are not part of
the \UniMath language either.
  
In \UniMath, 
types are stratified
according to their \emph{homotopy level}: we say that a type is
\fat{contractible} if it has exactly one element/inhabitant.  
A type is a \fat{proposition} if any two of its inhabitants are
identical
(there need
not be any inhabitant, corresponding to an unprovable proposition).  A
type is a \fat{set} if all of its identity types are
propositions. 
The hierarchy of homotopy levels continues with groupoids, 2-groupoids
and so on, but in the present work these higher levels are not used.
  
We call propositional truncation a type transformation that associates
to any type $A$ the proposition $\brck{A}$.  Intuitively,
$\brck{A}$ is empty when $A$ is, and contractible otherwise.  Note
that propositional truncation, often implemented as a HIT, is
implemented in \UniMath via a universal quantification, in the style of a 
generalized double negation:
\[ \brck{A} := \prd{P : \prop} (A \to P) \to P \enspace . \]

The propositional truncation is used to turn the strong,
constructive, existential quantification into a weak, classical, one:
we write $\exists a : A, B(a)$ for $\brck{\Sigma_{a:A}B(a)}$.
As in~\cite{hottbook} we use the term \fat{merely exists} to denote
the weaker notion of existence.

This distinction between structure and property given by the two
different existential quantifiers---$\Sigma$ and $\exists$, respectively---is
also reflected in our use of the vocabulary
`Problem \& Construction' vs.\ `Theorem \& Proof'.
Indeed, whenever
we describe the construction of a structure, that is,
when we construct a term of a type that is not a proposition in the above sense,
we use the terminology `Problem \& Construction'.
The pair `Theorem \& Proof' is reserved for the construction of inhabitants of 
a proposition. A corner case is strong unique existence which is exactly the same as being contractible, and contractibility of a given type is a proposition, but still it comes with a construction.

The \UniMath library contains a significant amount of category theory, for details see
\cite{rezk_completion}. A \fat{category} $\C$ in \UniMath is given by:
\begin{itemize}
  \item a type $\C_0$ of objects;
  \item for any two objects $A,B : \C_0$, a \fat{set} $\C(A,B)$ of
    morphisms;
  \item for any three objects $A,B,C : \C_0$, a composition operation
    \[ \compose{\_}{\_} : \C(B,C)\to \C(A,B)\to \C(A,C)\] 
  \item for any object $A : \C_0$, an identity arrow
    $1 = 1_A : \C(A,A)$,
\end{itemize}
subject to the usual axioms of category theory. Functors, natural
transformations, etc.\ are defined in the usual way.

The category $\set$ has as objects sets and as morphisms from $X$ to
$Y$ the set of~(type-theoretic) functions from $X$ to $Y$. Given
categories $\C$ and $\D$, we denote by $[\C,\D]$ the category of
functors from $\C$ to $\D$, and natural transformations between them.

  The article \cite{rezk_completion} calls \enquote{precategory} the notion here introduced as category,
   and reserves the word \enquote{category} for precategories with an additional property, called \enquote{univalence (for categories)}. 
   This property is not relevant for the work reported here. 
   We will occasionally remark on what would be guaranteed in addition for a univalent base category.
   The category $\set$ is univalent, and univalence is inherited from
   the target category $\D$ of a functor category $[\C,\D]$.

In the present work we reuse the existing library of category theory
and extend it as described below.

\subsection{Notational conventions regarding category
  theory}\label{sec:notation_category_theory}

We assume the reader to be familiar with the concepts of category
theory.  Here, we only point to the specific but rather standard
notations and conventions we will use throughout.

Instead of writing that $F$ is an object of the functor category
$[\C,\D]$, we often abbreviate this to $F:[\C,\D]$, but also to
$F:\C\to\D$. Given $d : \D$, we call $\constfunctor{d} : \C \to \D$
the functor that is constantly $d$ and $1_d$ on objects and morphisms,
respectively. This notation hides the category $\C$, which will
usually be deducible from the context. We write $\Id{\C}$ for the
identity endofunctor on $\C$. We also let \emph{(co)product} denote
general indexed (co)products and explicitly write if they are binary.

The category $\Ptd(\C)$ has, as objects, pointed endofunctors on $\C$,
that is, pairs of an endofunctor $F : \C \to \C$ and a natural
transformation $\eta : \Id{\C} \to F$.  We write $\idwt{\Ptd(\C)}$ for
the identity functor with its trivial point. Let $U$ be the forgetful
functor from $\Ptd(\C)$ to $[\C,\C]$ (that forgets the point).

Categories, functors and natural transformations constitute the prime
example of a 2-category. We write $\vcomp$ for vertical composition of
natural transformations and $\hcomp$ for their horizontal
composition. If one of the arguments to horizontal composition is the
identity on some functor, we just write the functor as the respective
argument. The corner case where both arguments are the identity on
some functors $X$ and $Y$ is just functor composition that is hence
written $X\hcomp Y$ (on objects and morphisms, this is $X$ applied
after $Y$, hence $(X\hcomp Y)(A)=X(YA)$ and likewise for
morphisms). Horizontal composition of $\mu:F\to G$ and $\nu:F'\to G'$
has $\mu\hcomp\nu:F\hcomp F'\to G\hcomp G'$ provided $F,G:\D\to\E$ and
$F',G':\C\to\D$.  The order of vertical composition $\vcomp$ is the
same as of functor composition: if $F,G,H:\C\to\D$ and $\mu:G\to H$
and $\nu:F\to G$, then $\mu\vcomp\nu:F\to H$ is defined by object-wise
composition in $\D$.

Given a functor $F : [\A,\B]$ and a category $\C$ we define the
functor $\_ \hcomp F$ on functor categories:
\[
 \_ \hcomp F : [\B,\C] \to [\A,\C]
\]
This functor takes a functor $X : [\B,\C]$ and precomposes it with
$F$, that is, $X \mapsto X\hcomp F$, and likewise with the morphisms, 
i.\,e., the natural transformations. Once again the category $\C$ is
hidden, but it can often be deduced from the context.

We follow \cite{DBLP:journals/corr/AhrensM16} in making explicit the
monoidal structure on functor category $[\C,\C]$ that carries over to
$\Ptd(\C)$: let
$\alpha_{X,Y,Z} : X \hcomp (Y \hcomp Z) \simeq (X \hcomp Y) \hcomp Z$,
$\rho_X : 1_{\C} \hcomp X \simeq X$ and
$\lambda_X : X \hcomp 1_{\C} \simeq X$ denote the monoidal
isomorphisms. Notice that all those morphisms are pointwise the
identity, but making them explicit is needed for typechecking in the
implementation \cite{DBLP:journals/corr/AhrensM16}.

\section{Two notions of signature}\label{sec:signature}

As outlined in the introduction, a signature abstractly
specifies a datatype by describing the shape of 
elements of that type.
We give two notions of signatures suitable for the description
of languages with variable binding, such as the untyped
lambda calculus.
We first describe a rather syntactic notion of signature:
binding signatures. We then proceed
with a description of a semantic notion of signature: signatures
with strength. We give constructions to obtain
signatures with strength and finally associate a signature with
strength to each binding signature.

\subsection{Binding signatures}\label{subsec:bindingsignatures}

A \emph{binding signature} is given by simple syntactic data 
that allows one to concisely 
specify a language with variable binding.
Binding signatures are less expressive than the signatures with strength that
will be presented in the next section. On the other hand, they
are easier to specify.

\begin{definition}[Arity, Binding signature]\label{def:binding_sig}
  An \fat{arity} is a (finite) list of natural numbers. A \fat{binding
    signature} is a family of arities, more precisely,
  \begin{itemize}
    \item a type $I$ with decidable equality and
    \item a function $\arity : I \to \lst{\nat}$.
  \end{itemize}
\end{definition}

Intuitively, the type $I$ of a binding signature indexes the language
constructors, and the function $\arity$ associates an arity to each
constructor.
We need decidable equality on the indexing type (which, by Hedberg's
theorem \cite{DBLP:journals/jfp/Hedberg98}, is a set)
in the proof of Lemma~\ref{lem:sem_sig_cocont}. 
Hypothesizing a
decidable equality also makes our notion of binding signature equal to
the notion used in \cite{fpt}.

In \UniMath we define this as a nested $\Sigma$-type (with \coqdocdefinition{UU} for the universe $\type$):

\begin{coqdoccode}
\coqdocemptyline
\coqdocnoindent
\coqdockw{Definition} \coqdef{UniMath.SubstitutionSystems.BindingSigToMonad.BindingSig}{BindingSig}{\coqdocdefinition{BindingSig}} : \coqref{UniMath.Foundations.Basics.Preamble.UU}{\coqdocdefinition{UU}} := \coqref{UniMath.Foundations.Basics.Preamble.:type scope:'xCExA3' x '..' x ',' x}{\coqdocnotation{Σ}} \coqref{UniMath.Foundations.Basics.Preamble.:type scope:'xCExA3' x '..' x ',' x}{\coqdocnotation{(}}\coqdocvar{I} : \coqref{UniMath.Foundations.Basics.Preamble.UU}{\coqdocdefinition{UU}}) (\coqdocvar{h} : \coqref{UniMath.Foundations.Basics.PartB.isdeceq}{\coqdocdefinition{isdeceq}} \coqdocvariable{I}\coqref{UniMath.Foundations.Basics.Preamble.:type scope:'xCExA3' x '..' x ',' x}{\coqdocnotation{),}} \coqdocvariable{I} \coqexternalref{:type scope:x '->' x}{http://coq.inria.fr/distrib/8.5pl3/stdlib/Coq.Init.Logic}{\coqdocnotation{\ensuremath{\rightarrow}}} \coqref{UniMath.Foundations.Combinatorics.Lists.list}{\coqdocdefinition{list}} \coqexternalref{nat}{http://coq.inria.fr/distrib/8.5pl3/stdlib/Coq.Init.Datatypes}{\coqdocinductive{nat}}.\coqdoceol
\coqdocemptyline
\end{coqdoccode}

We also define functions for accessing the components of a
\coqdocdefinition{BindingSig}
and a constructor function for constructing one:

\begin{coqdoccode}
\coqdocemptyline
\coqdocnoindent
\coqdockw{Definition} \coqdef{UniMath.SubstitutionSystems.BindingSigToMonad.BindingSigIndex}{BindingSigIndex}{\coqdocdefinition{BindingSigIndex}} : \coqref{UniMath.SubstitutionSystems.BindingSigToMonad.BindingSig}{\coqdocdefinition{BindingSig}} \coqexternalref{:type scope:x '->' x}{http://coq.inria.fr/distrib/8.5pl3/stdlib/Coq.Init.Logic}{\coqdocnotation{\ensuremath{\rightarrow}}} \coqref{UniMath.Foundations.Basics.Preamble.UU}{\coqdocdefinition{UU}} := \coqref{UniMath.Foundations.Basics.Preamble.pr1}{\coqdocdefinition{pr1}}.\coqdoceol
\coqdocnoindent
\coqdockw{Definition} \coqdef{UniMath.SubstitutionSystems.BindingSigToMonad.BindingSigIsdeceq}{BindingSigIsdeceq}{\coqdocdefinition{BindingSigIsdeceq}} (\coqdocvar{s} : \coqref{UniMath.SubstitutionSystems.BindingSigToMonad.BindingSig}{\coqdocdefinition{BindingSig}}) : \coqref{UniMath.Foundations.Basics.PartB.isdeceq}{\coqdocdefinition{isdeceq}} (\coqref{UniMath.SubstitutionSystems.BindingSigToMonad.BindingSigIndex}{\coqdocdefinition{BindingSigIndex}} \coqdocvariable{s}) :=\coqdoceol
\coqdocindent{1.00em}
\coqref{UniMath.Foundations.Basics.Preamble.pr1}{\coqdocdefinition{pr1}} (\coqref{UniMath.Foundations.Basics.Preamble.pr2}{\coqdocdefinition{pr2}} \coqdocvariable{s}).\coqdoceol
\coqdocnoindent
\coqdockw{Definition} \coqdef{UniMath.SubstitutionSystems.BindingSigToMonad.BindingSigMap}{BindingSigMap}{\coqdocdefinition{BindingSigMap}} (\coqdocvar{s} : \coqref{UniMath.SubstitutionSystems.BindingSigToMonad.BindingSig}{\coqdocdefinition{BindingSig}}) : \coqref{UniMath.SubstitutionSystems.BindingSigToMonad.BindingSigIndex}{\coqdocdefinition{BindingSigIndex}} \coqdocvariable{s} \coqexternalref{:type scope:x '->' x}{http://coq.inria.fr/distrib/8.5pl3/stdlib/Coq.Init.Logic}{\coqdocnotation{\ensuremath{\rightarrow}}} \coqref{UniMath.Foundations.Combinatorics.Lists.list}{\coqdocdefinition{list}} \coqexternalref{nat}{http://coq.inria.fr/distrib/8.5pl3/stdlib/Coq.Init.Datatypes}{\coqdocinductive{nat}} :=\coqdoceol
\coqdocindent{1.00em}
\coqref{UniMath.Foundations.Basics.Preamble.pr2}{\coqdocdefinition{pr2}} (\coqref{UniMath.Foundations.Basics.Preamble.pr2}{\coqdocdefinition{pr2}} \coqdocvariable{s}).\coqdoceol
\coqdocemptyline
\coqdocnoindent
\coqdockw{Definition} \coqdef{UniMath.SubstitutionSystems.BindingSigToMonad.mkBindingSig}{mkBindingSig}{\coqdocdefinition{mkBindingSig}} \{\coqdocvar{I} : \coqref{UniMath.Foundations.Basics.Preamble.UU}{\coqdocdefinition{UU}}\} (\coqdocvar{h} : \coqref{UniMath.Foundations.Basics.PartB.isdeceq}{\coqdocdefinition{isdeceq}} \coqdocvariable{I}) (\coqdocvar{f} : \coqdocvariable{I} \coqexternalref{:type scope:x '->' x}{http://coq.inria.fr/distrib/8.5pl3/stdlib/Coq.Init.Logic}{\coqdocnotation{\ensuremath{\rightarrow}}} \coqref{UniMath.Foundations.Combinatorics.Lists.list}{\coqdocdefinition{list}} \coqexternalref{nat}{http://coq.inria.fr/distrib/8.5pl3/stdlib/Coq.Init.Datatypes}{\coqdocinductive{nat}}) : \coqref{UniMath.SubstitutionSystems.BindingSigToMonad.BindingSig}{\coqdocdefinition{BindingSig}} := \coqdoceol
\coqdocindent{1.00em}
 (\coqdocvariable{I}\coqref{UniMath.Foundations.Basics.Preamble.::x ',,' x}{\coqdocnotation{,,}}(\coqdocvariable{h}\coqref{UniMath.Foundations.Basics.Preamble.::x ',,' x}{\coqdocnotation{,,}}\coqdocvariable{f})).\coqdoceol
\coqdocemptyline
\end{coqdoccode}

This way we can mimic the behavior of \Coq's \texttt{Record} types
which are not part of \UniMath as they are defined using
\texttt{Inductive}.

We can take the coproduct of two binding signatures by taking the
coproduct of the underlying indexing sets, and, for the function
specifying the arities, the induced function on the coproduct type.

\begin{example}[Binding signature of untyped lambda calculus]\label{ex:binding_sig_lc}
  The binding signature of the untyped lambda calculus is given by
  $I:= \{ \abs{}, \app{} \}$ and the arity function is
  \[  \abs{} \mapsto [1] \enspace , \enspace \app{} \mapsto [0,0] \enspace .\] 
This is to be read as follows: there are---besides variables that are
treated generically in Section~\ref{sec:hss}---two constructors. 
The first constructor $\abs{}$, corresponding to lambda abstraction,
has just one argument (as $\arity(\abs{})$ is a one-element list), 
and this argument
can make use of 1 extra variable being bound by the constructor. 
The second constructor $\app{}$, corresponding to application, has
two arguments, and there is no binding involved.
\end{example}

\begin{example}[Binding signature of presyntax of Martin-Löf type theory]\label{ex:binding_sig_mltt}
  The binding signature of Martin-Löf type theory is given 
  in Section~\ref{sec:monad}, as part of an extended example that uses an infinite index set.
  Using the coproduct of binding signatures, it can easily be
  decomposed, in particular, using the binding signature of the untyped
  lambda calculus as one ingredient. 
\end{example}

\subsection{Signatures with strength}\label{subsec:signatureswithstrength}

The next, more semantic, notion of signature was defined in
\cite[Definition~5]{DBLP:journals/tcs/MatthesU04}; there, it was merely
called \enquote{signature}. In order to explicitly distinguish them
from binding signatures, we call them \enquote{signatures with
  strength} here. As a new contribution, we organize
the signatures with strength as a category.

\begin{definition}[Signatures with strength]\label{def:sig_w_strength}

  Given a category $\C$, a \fat{signature with strength} is a pair
  $(H,\theta)$ of an endofunctor $H$ on $[\C,\C]$, called the
  signature functor, and a natural transformation $\th{}: (H {-})
  \hcomp U {\sim} \arr H ({-} \hcomp U {\sim})$ 
    between bifunctors $[\C,\C] \times
  \Ptd(\C) \arr [\C,\C]$ such that $\th{}$ is `linear' in the second
  component. 
\end{definition}

  In more detail, 
  the bifunctors applied to a pair of objects $(X,(Z,e))$ with $X:[\C,\C]$ and $(Z,e):\Ptd(\C)$ 
  ($X$ for the argument symbolized by $-$ and $(Z,e)$ for the argument symbolized by $\sim$) 
  yield $HX\hcomp Z$ and $H(X\hcomp Z)$, thus $\th{X,(Z,e)}:HX\hcomp Z\to H(X\hcomp Z)$.
  By `linearity' of $\th{}$ in the second argument we mean the equations
  \[ \th{X, \idwt{\Ptd(\C)}} = H(\lambda^{-1}_X) \vcomp \lambda_{HX} \] 
  (note that $\lambda_{HX}:HX\hcomp 1\to HX$ and $H(\lambda^{-1}_X):HX\to H(X\hcomp1)$, using the monoidal isomorphism $\lambda$ introduced in Section~\ref{sec:notation_category_theory}) and
  \[  \th{X, (Z' \hcomp Z, e' \hcomp e)} = H(\alpha^{-1}_{X,Z',Z}) \vcomp   \th{X \hcomp Z', (Z,
    e)} \vcomp (\th{X, (Z', e')} \hcomp Z) \vcomp \alpha_{HX, Z', Z} \enspace , \]
as illustrated by the diagram
  \[
  \begin{xy}
   \xymatrix@C=50pt@R=25pt{
                  **[l]H X \hcomp (Z' \hcomp Z)  \ar[rr]^{\th{X,(Z' \hcomp Z, e' \hcomp e)}} \ar[d]_{\alpha_{HX, Z', Z}} & &  **[r]H(X \hcomp (Z' \hcomp Z)) \\
                  **[l] (H X \hcomp Z') \hcomp Z \ar[r]^-{\th{X, (Z', e')} \hcomp Z} &   **[c]H(X \hcomp Z')\hcomp Z \ar[r]^-{\th{X \hcomp Z', (Z,
    e)}} &   **[r]H((X \hcomp Z')\hcomp Z) \ar[u]_{H(\alpha^{-1}_{X,Z',Z})}\\
     }
  \end{xy}
  \]

\begin{definition}[Morphism of signatures with strength]\label{def:mor_sig_w_strength}
  Given two signatures with
  strength $(H,\theta)$ and $(H',\theta')$, a \fat{morphism of
    signatures with strength} from $(H,\theta)$ to $(H',\theta')$ is a natural
  transformation $h : H \to H'$ such that the following diagram
  commutes for any $X : [\C,\C]$ and $(Z,e) : \Ptd(\C)$.
  \[
  \begin{xy}
   \xymatrix@C=50pt@R=30pt{
                  **[l]H X \hcomp Z  \ar[r]^{\th{X,(Z,e)}} \ar[d]_{h_X \hcomp Z} &   **[r]H (X \hcomp Z) \ar[d]^{h_{X\hcomp Z}}\\
                  **[l]H' X \hcomp Z  \ar[r]^{\th{X,(Z,e)}'} &   **[r]H' (X \hcomp Z) \\
     }
  \end{xy}
  \]
\end{definition}

Composition and identity morphisms of signatures with strength are
given by composition and identity of natural transformations.  This
defines the category of signatures with strength.

Examples of signatures with strength are given in
\cite{DBLP:journals/tcs/MatthesU04}. Another way of producing examples
is the map defined in Construction~\ref{constr:sem_sig_from_binding}.

The signatures with strength do not distinguish between arities and
signatures.
As developed in \cite{DBLP:journals/corr/AhrensM16}, there is a way
to build a new signature by taking the coproduct of two signatures.
Intuitively, and just as for binding signatures, this corresponds to constructing a new language by taking
the disjoint union of the language constructors of two given
languages. 
 What is new here compared to \cite{DBLP:journals/corr/AhrensM16}
 is the explicitly categorical treatment (i.\,e., taking into account morphisms
 of signatures with strength). 
The construction generalizes easily to the coproduct of an arbitrary
family of such signatures:

\begin{definition}[Coproduct of signatures with strength]\label{def:coprod_sem_sig}
  If $\C$ has coproducts, then the coproduct of a family of signatures
  with strength is defined as follows:
   \begin{itemize}
     \item the signature functor is given by the coproduct in the
       endofunctor category on $[\C,\C]$ induced by that on $\C$;
     \item the strength is induced by coproduct of arrows.
   \end{itemize}
   The strength laws are simple consequences of the strength laws of
   each member of the family of signatures, and the universal property is readily established.
\end{definition}

\begin{definition}[Binary product of signatures with strength]\label{def:prod_sem_sig}
  If $\C$ has binary products, then the binary product of two
  signatures with strength has, as signature functor, the binary
  product of the functors of the given signatures. 
  The strength is then induced analogously to coproducts.
\end{definition}
By way of iteration, binary products will be used to model multiple arguments of a datatype constructor.

Definitions~\ref{def:coprod_sem_sig} and \ref{def:prod_sem_sig} entail
that the forgetful functor from signatures with strength to
endofunctors on $[\C,\C]$ lifts and preserves coproducts and binary
products.

\subsection{Signatures with strength from binding signatures}
\label{subsec:sem_sig_from_binding}

Constructing suitable signatures with strength for a language
seems like a daunting
task.  Fortunately, it is often sufficient to specify the binding
signature.  The generic solution to the following problem then yields
the corresponding signature with strength.

\begin{problem}\label{prob:sem_sig_from_binding}
  Let $\C$ be a category with coproducts, binary products and
  a terminal object.
  Given a binding signature, to construct a signature with strength on $\C$.  
  This task is naturally divided into
  \begin{enumerate}
    \item the construction of the signature functor $H$ as endofunctor on
      $[\C,\C]$ and then\label{constr_item_H}
    \item the construction of a strength for $H$.\label{constr_item_theta}
  \end{enumerate}
\end{problem}

\begin{construction}[Part \ref{constr_item_H} of Problem~\ref{prob:sem_sig_from_binding}]\label{constr_H:sem_sig_from_binding}
  Let $(I,\arity)$ be a binding signature.
  Let $i:I$. To the list $\arity(i) = [n_1,\ldots,n_k]$ we associate the
  functor defined on objects by
  \[A \mapsto \prod_{1\leqslant j \leqslant k} (X \hcomp \option^{n_j})(A) \]
  Here, the functor $\option : \C\to\C$ is
  defined on objects by $\option(A) := 1 + A$.  The product is
  implemented as an iterated binary product.
   Put differently, we define a functor
   \begin{align*}
              [\C,\C] &\to [\C,\C] \\
                 X    &\mapsto \prod_{1 \leqslant j \leqslant k} X \hcomp \option^{n_j}
   \end{align*}
     The functor associated
   to the signature $(I,\arity)$ is then obtained as the coproduct of the
   functors associated to each arity,
   \begin{align*}
             H : [\C,\C] &\to [\C,\C] \\
                 X    &\mapsto \coprod_{i : I} \prod_{1 \leqslant j \leqslant \cfont{length}(\arity(i))} X \hcomp \option^{\arity(i)_j}
   \end{align*}
   
   For the construction of this functor over the category of sets (i.\,e., when $\C=\set$), it
   is essential for $I$ to be a \fat{set}.  This is the case, as a
   consequence of our hypothesis of $I$ having decidable equality as
   the indexing set of a binding signature.
 \end{construction}

 As we have just seen, the signature functors $H$ that arise from
 binding signatures are of a special shape, where the argument $X$
 only enters in the form of $X\hcomp \option^n$. This can be exploited
 in the construction of the strength
   $\theta$ for
 $H$. The right level of generality of this pattern is signature
 functors $H$ that are given by precomposition with a fixed
 endofunctor $G$ on $\C$, i.\,e., with $HX=X\hcomp G$. Pointed
 distributive laws for $G$ to be introduced next will lift to
 strengths for $H$, hence providing signatures with strength from a
 simpler input.  
 \begin{definition}[Pointed distributive law] 
   Let $\C$ be a category and $G:[\C,\C]$. A pointed distributive law
   for $G$ is a natural transformation $\delta: G\hcomp U{\sim}\to
   U{\sim}\hcomp G$ of functors $\Ptd(\C)\to[\C,\C]$ such that
  \[\delta_\idwt{\Ptd(\C)} =1_G\]
  and
  \[\delta_{(Z' \hcomp Z, e' \hcomp e)}= \alpha_{Z', Z,G}\vcomp Z'\hcomp\delta_{(Z,e)}\vcomp \alpha^{-1}_{Z',G,Z}\vcomp\delta_{(Z',e')}\hcomp Z\vcomp\alpha_{G,Z', Z}\enspace,\]
where the second equation is commutation of the following diagram:
  \[
  \begin{xy}
   \xymatrix@C=35pt@R=25pt{
                  **[l]G\hcomp (Z' \hcomp Z)  \ar[rrr]^{\delta_{(Z' \hcomp Z, e' \hcomp e)}} \ar[d]_{\alpha_{G, Z', Z}} & & &  **[r](Z' \hcomp Z)\hcomp G \\
                  **[l] (G \hcomp Z') \hcomp Z \ar[r]^-{\delta_{(Z', e')} \hcomp Z} &   **[c](Z'\hcomp G)\hcomp Z \ar[r]^{\alpha^{-1}_{Z',G,Z}} &   **[c]Z'\hcomp(G \hcomp Z) \ar[r]^-{Z'\hcomp\delta_{(Z,e)}}&   **[r]Z'\hcomp(Z \hcomp G) \ar[u]_{\alpha_{Z', Z,G}}\\
     }
  \end{xy}
  \]
\end{definition}
Note that, in analogy with the definition of signature with strength, we symbolize the sole argument of the functors as $\sim$.
Note that setting $\delta_\idwt{\Ptd(\C)}$ to
$\rho^{-1}_G\vcomp\lambda_G$ instead of the identity would be to emphasize
the monoidal structure on $[\C,\C]$, but our implementation did not
run into problems with our simplified definition (that, anyway, is pointwise identical).

The prime example is with $G=\option$, where
\[\delta_{(Z,e)}(A)=[e_{\option(A)}\vcomp\inl_{1,A} \; , \; Z(\inr_{1,A})]:\option(ZA)\to Z(\option(A))\enspace, \] 
with the injections $\inl_{1,A}$ and $\inr_{1,A}$ into $\option(A)$.

The following lemma is obtained by easy calculations.
\begin{lemma}\label{lem:strength_from_distrib_law}
  Let $\C$ be a category, $G:[\C,\C]$ and $\delta$ a pointed
  distributive law for $G$. Let $H$ be precomposition with $G$, then
\[\th{X, (Z, e)}:=\alpha_{X, Z,G}\vcomp X\hcomp\delta_{(Z,e)}\vcomp\alpha^{-1}_{X,G,Z}\enspace,\] 
as illustrated by the diagram
\[
  \begin{xy}
   \xymatrix@C=50pt@R=25pt{
                  **[l](X\hcomp G) \hcomp Z  \ar[r]^{\th{X,(Z,e)}} \ar[d]_{\alpha^{-1}_{X,G,Z}} &   **[r](X \hcomp Z)\hcomp G \\
                  **[l]X \hcomp(G \hcomp Z)  \ar[r]^{X\hcomp\delta{(Z,e)}} &   **[r]X \hcomp (Z\hcomp G) \ar[u]_{\alpha_{X, Z,G}}\\
     }
  \end{xy}
  \]
yields a natural transformation, and $(H,\theta)$ is a signature with strength.
\end{lemma}

Also the next lemma is obtained by easy calculations.
\begin{lemma}\label{lem:distrib_law_comp}
Let $\C$ be a category, $G_1, G_2:[\C,\C]$ with pointed distributive
laws $\delta^1$ and $\delta^2$, respectively. Then, the following is a pointed distributive law for
$G_1\hcomp G_2$:
\[\delta_{(Z,e)}:=\alpha^{-1}_{Z,G_1,G_2}\vcomp\delta^1_{(Z,e)}\hcomp G_2\vcomp \alpha_{G_1,Z,G_2}\vcomp G_1\hcomp\delta^2_{(Z,e)}\vcomp\alpha^{-1}_{G_1,G_2,Z}\enspace,\]
visualized as follows: 
\[
  \begin{xy}
   \xymatrix@C=30pt@R=25pt{
                  **[l](G_1\hcomp G_2) \hcomp Z  \ar[rrr]^{\delta_{(Z, e)}} \ar[d]_{\alpha^{-1}_{G_1,G_2,Z}} & & &  **[r]Z\hcomp (G_1\hcomp G_2) \\
                  **[l] G_1\hcomp (G_2 \hcomp Z) \ar[r]^-{G_1\hcomp\delta^2_{(Z, e)}} &   **[c]G_1\hcomp (Z\hcomp G_2) \ar[r]^{\alpha_{G_1,Z,G_2}} &   **[c] (G_1\hcomp Z)\hcomp G_2\ar[r]^-{\delta^1_{(Z,e)}\hcomp G_2}&   **[r](Z\hcomp G_1)\hcomp G_2 \ar[u]_{\alpha^{-1}_{Z,G_1,G_2}}\\
     }
  \end{xy}
  \]
\end{lemma}

\begin{construction}[Part \ref{constr_item_theta} of Problem~\ref{prob:sem_sig_from_binding}]\label{constr:sem_sig_from_binding}
  Let $(I,\arity)$ be a binding signature.
  It suffices to define the signature with strength associated to any
  $\arity(i)$ for $i : I$.  The signature with strength associated to
  $(I,\arity)$ is then obtained by taking the coproduct of all the
  signatures with strength associated to $\arity(i)$ as in
  Definition~\ref{def:coprod_sem_sig}.

  Let $i:I$. Thanks to Definition~\ref{def:prod_sem_sig} for binary
  products, used repeatedly in order to account for multiple arguments
  (i.\,e., multiple elements in the list $\arity(i)$), it suffices to
  define the strength associated to the endofunctor on $[\C,\C]$,
  expressed by the term $X\hcomp\option^{n_k}$ in the above
  construction. However, this is an instance of
  Lemma~\ref{lem:strength_from_distrib_law}, with $G=\option^{n_k}$,
  and the latter is an iterated composition of $\option$ for which the
  pointed distributive law has been given above. So, Lemma~\ref{lem:distrib_law_comp}
  provides a pointed distributive law for
  $\option^{n_k}$.
\end{construction}

\begin{example}[The signature with strength for the untyped lambda calculus]\label{ex:sig_strength_lc}
  Consider the binding signature of Example~\ref{ex:binding_sig_lc}.
  The signature functor obtained from that binding signature via the
  map defined in Construction~\ref{constr_H:sem_sig_from_binding} is
  given by
  \[ X \mapsto X \hcomp \option + X\times X \]
  We also obtain a strength law for this functor by
  Construction~\ref{constr:sem_sig_from_binding}. For more details
  about this see~\cite{DBLP:journals/tcs/MatthesU04}.
\end{example}

The next section 
is dedicated to the construction of
initial algebras for the signature functor associated to a binding signature by Construction~\ref{constr_H:sem_sig_from_binding},
culminating in Theorem~\ref{lem:sem_sig_cocont} and Construction~\ref{thm:initial_alg_for_binding_sig}. 
In Section~\ref{sec:thewholechain} we then equip those initial algebras with a monad structure.

\section{Construction of datatypes as initial algebras}\label{sec:inductive_types}

Given a category $\D$, we define the datatype
specified by a functor $F : \D \to \D$ to be any initial algebra
of $F$. 
Note that by this definition, such datatypes are only defined 
up to unique isomorphism.
For a given endofunctor $F$ on $\D$, an initial algebra might or might not exist. 
In this section, we construct initial algebras for signature functors as in
Section~\ref{sec:signature}, with $\D$ instantiated to the category of
endofunctors on the category of sets, hence with category $\C$ of the previous section fixed to $\set$. 
However, the results of this section are stated and proved for arbitrary categories $\C$ equipped with suitable
structure, and only instantiated to $\set$ in the end.

Our main tool for the construction of initial algebras is Construction~\ref{thm:init_alg}. 
That construction yields an initial $F$-algebra for an $\omega$-cocontinuous endofunctor $F$
from a certain colimit.
It hence reduces our task of constructing  datatypes
(i.\,e., initial
algebras) to the construction of certain colimits (see
Section~\ref{sec:colimits_in_sets}) and to showing that various
functors preserve these colimits 
(see Sections~\ref{sec:inductive_sets} and
\ref{sec:inductive_families}).

\subsection{Colimits}

In our formalization, colimits are parametrized by diagrams over
graphs, as suggested by \cite[p.~71]{maclane}.
\begin{definition}[Graph]
   A \fat{graph} is a pair consisting of
   \begin{itemize}
     \item a type $\vertex : \UU$ representing the vertices and
     \item a family $\edge : \vertex \to \vertex \to \UU$ representing
       the edges as a dependent family of types.
   \end{itemize}
\end{definition}
A diagram, accordingly, is a map from a graph into the graph
underlying a category.
\begin{definition}[Diagram]
  Given a graph $G$ made of $\vertex_G$ and $\edge_G$ and a category $\C$, a \fat{diagram} of shape
  $G$ in $\C$ is a pair consisting of
  \begin{itemize}
    \item a map $\dob : \vertex_G\to \C_0$ and
    \item a family of maps $\dmor : \prd{u,v:\vertex_G}\edge_G(u,v) \to
      \C(\dob(u),\dob(v))$.
  \end{itemize}
\end{definition}
Henceforth, we will abbreviate $u:\vertex_G$ by $u:G$. These
definitions are also conveniently represented in \UniMath using
$\Sigma$-types with suitable accessor and constructor functions:
\newcommand{\leftdoublebracket}{[\kern-0.35em[}
\newcommand{\rightdoublebracket}{]\kern-0.35em]}
\begin{coqdoccode}
\coqdocemptyline
\coqdocnoindent
\coqdockw{Definition} \coqdef{UniMath.CategoryTheory.limits.graphs.colimits.graph}{graph}{\coqdocdefinition{graph}} := \coqref{UniMath.Foundations.Basics.Preamble.:type scope:'xCExA3' x '..' x ',' x}{\coqdocnotation{Σ}} \coqref{UniMath.Foundations.Basics.Preamble.:type scope:'xCExA3' x '..' x ',' x}{\coqdocnotation{(}}\coqdocvar{D} : \coqref{UniMath.Foundations.Basics.Preamble.UU}{\coqdocdefinition{UU}}\coqref{UniMath.Foundations.Basics.Preamble.:type scope:'xCExA3' x '..' x ',' x}{\coqdocnotation{),}} \coqdocvariable{D} \coqexternalref{:type scope:x '->' x}{http://coq.inria.fr/distrib/8.5pl3/stdlib/Coq.Init.Logic}{\coqdocnotation{\ensuremath{\rightarrow}}} \coqdocvariable{D} \coqexternalref{:type scope:x '->' x}{http://coq.inria.fr/distrib/8.5pl3/stdlib/Coq.Init.Logic}{\coqdocnotation{\ensuremath{\rightarrow}}} \coqref{UniMath.Foundations.Basics.Preamble.UU}{\coqdocdefinition{UU}}.\coqdoceol
\coqdocemptyline
\coqdocnoindent
\coqdockw{Definition} \coqdef{UniMath.CategoryTheory.limits.graphs.colimits.vertex}{vertex}{\coqdocdefinition{vertex}} : \coqref{UniMath.CategoryTheory.limits.graphs.colimits.graph}{\coqdocdefinition{graph}} \coqexternalref{:type scope:x '->' x}{http://coq.inria.fr/distrib/8.5pl3/stdlib/Coq.Init.Logic}{\coqdocnotation{\ensuremath{\rightarrow}}} \coqref{UniMath.Foundations.Basics.Preamble.UU}{\coqdocdefinition{UU}} := \coqref{UniMath.Foundations.Basics.Preamble.pr1}{\coqdocdefinition{pr1}}.\coqdoceol
\coqdocnoindent
\coqdockw{Definition} \coqdef{UniMath.CategoryTheory.limits.graphs.colimits.edge}{edge}{\coqdocdefinition{edge}} \{\coqdocvar{g} : \coqref{UniMath.CategoryTheory.limits.graphs.colimits.graph}{\coqdocdefinition{graph}}\} : \coqref{UniMath.CategoryTheory.limits.graphs.colimits.vertex}{\coqdocdefinition{vertex}} \coqdocvariable{g} \coqexternalref{:type scope:x '->' x}{http://coq.inria.fr/distrib/8.5pl3/stdlib/Coq.Init.Logic}{\coqdocnotation{\ensuremath{\rightarrow}}} \coqref{UniMath.CategoryTheory.limits.graphs.colimits.vertex}{\coqdocdefinition{vertex}} \coqdocvariable{g} \coqexternalref{:type scope:x '->' x}{http://coq.inria.fr/distrib/8.5pl3/stdlib/Coq.Init.Logic}{\coqdocnotation{\ensuremath{\rightarrow}}} \coqref{UniMath.Foundations.Basics.Preamble.UU}{\coqdocdefinition{UU}} := \coqref{UniMath.Foundations.Basics.Preamble.pr2}{\coqdocdefinition{pr2}} \coqdocvariable{g}.\coqdoceol
\coqdocnoindent
\coqdockw{Definition} \coqdef{UniMath.CategoryTheory.limits.graphs.colimits.mk graph}{mk\_graph}{\coqdocdefinition{mk\_graph}} (\coqdocvar{D} : \coqref{UniMath.Foundations.Basics.Preamble.UU}{\coqdocdefinition{UU}}) (\coqdocvar{e} : \coqdocvariable{D} \coqref{UniMath.Foundations.Basics.Preamble.:type scope:x 'xE2x86x92' x}{\coqdocnotation{→}} \coqdocvariable{D} \coqref{UniMath.Foundations.Basics.Preamble.:type scope:x 'xE2x86x92' x}{\coqdocnotation{→}} \coqref{UniMath.Foundations.Basics.Preamble.UU}{\coqdocdefinition{UU}}) : \coqref{UniMath.CategoryTheory.limits.graphs.colimits.graph}{\coqdocdefinition{graph}} := \coqref{UniMath.Foundations.Basics.Preamble.tpair}{\coqdocconstructor{tpair}} \coqdocvar{\_} \coqdocvariable{D} \coqdocvariable{e}.\coqdoceol
\coqdocemptyline
\coqdocnoindent
\coqdockw{Definition} \coqdef{UniMath.CategoryTheory.limits.graphs.colimits.diagram}{diagram}{\coqdocdefinition{diagram}} (\coqdocvar{g} : \coqref{UniMath.CategoryTheory.limits.graphs.colimits.graph}{\coqdocdefinition{graph}}) (\coqdocvar{C} : \coqref{UniMath.CategoryTheory.precategories.precategory}{\coqdocdefinition{precategory}}) : \coqref{UniMath.Foundations.Basics.Preamble.UU}{\coqdocdefinition{UU}} :=\coqdoceol
\coqdocindent{1.00em}
\coqref{UniMath.Foundations.Basics.Preamble.:type scope:'xCExA3' x '..' x ',' x}{\coqdocnotation{Σ}} \coqref{UniMath.Foundations.Basics.Preamble.:type scope:'xCExA3' x '..' x ',' x}{\coqdocnotation{(}}\coqdocvar{f} : \coqref{UniMath.CategoryTheory.limits.graphs.colimits.vertex}{\coqdocdefinition{vertex}} \coqdocvariable{g} \coqexternalref{:type scope:x '->' x}{http://coq.inria.fr/distrib/8.5pl3/stdlib/Coq.Init.Logic}{\coqdocnotation{\ensuremath{\rightarrow}}} \coqdocvariable{C}\coqref{UniMath.Foundations.Basics.Preamble.:type scope:'xCExA3' x '..' x ',' x}{\coqdocnotation{),}} \coqref{UniMath.Foundations.Basics.Preamble.:type scope:'xCExA0' x '..' x ',' x}{\coqdocnotation{\ensuremath{\Pi}}} \coqref{UniMath.Foundations.Basics.Preamble.:type scope:'xCExA0' x '..' x ',' x}{\coqdocnotation{(}}\coqdocvar{a} \coqdocvar{b} : \coqref{UniMath.CategoryTheory.limits.graphs.colimits.vertex}{\coqdocdefinition{vertex}} \coqdocvariable{g}\coqref{UniMath.Foundations.Basics.Preamble.:type scope:'xCExA0' x '..' x ',' x}{\coqdocnotation{),}} \coqref{UniMath.CategoryTheory.limits.graphs.colimits.edge}{\coqdocdefinition{edge}} \coqdocvariable{a} \coqdocvariable{b} \coqexternalref{:type scope:x '->' x}{http://coq.inria.fr/distrib/8.5pl3/stdlib/Coq.Init.Logic}{\coqdocnotation{\ensuremath{\rightarrow}}} \coqdocvariable{C}\coqref{UniMath.CategoryTheory.UnicodeNotations.::x 'xE2x9FxA6' x ',' x 'xE2x9FxA7'}{\coqdocnotation{\leftdoublebracket}}\coqdocvariable{f} \coqdocvariable{a}\coqref{UniMath.CategoryTheory.UnicodeNotations.::x 'xE2x9FxA6' x ',' x 'xE2x9FxA7'}{\coqdocnotation{,}} \coqdocvariable{f} \coqdocvariable{b}\coqref{UniMath.CategoryTheory.UnicodeNotations.::x 'xE2x9FxA6' x ',' x 'xE2x9FxA7'}{\coqdocnotation{\rightdoublebracket}}.\coqdoceol
\coqdocemptyline
\end{coqdoccode}

\begin{remark}
  For conceptual economy, it is customary in category theory to 
  index limits and colimits by categories instead of graphs, and by functors instead of
  diagrams. The extra structure
  that categories and functors have compared to graphs and diagrams is not
  used in what we are presenting here. However, our formalization
  can also be
  used with categories and functors, thanks to coercions from 
  categories and functors 
  to graphs and diagrams, respectively.
\end{remark}

\begin{definition}[Cocone]
  Given a diagram $d$ of shape $G$ in $\C$ made of $\dob_d$ and $\dmor_d$, and an object $C : \C_0$,
  a \fat{cocone} under $d$ with tip $C$ is given by
  \begin{itemize}
    \item a family of morphisms $a:\prd{v:G} \C(\dob_d(v),C)$
      and
    \item a family of equalities $\prd{u:G,v:G,e:\edge_G(u,v)}
      \compose{\dmor_d(e)}{a(v)} = a(u)$.
  \end{itemize}
  Let $\Cocone(d,C)$ be the type of cocones under $d$ with tip $C$.
\end{definition}
The equalities in the definition can be depicted as:
  
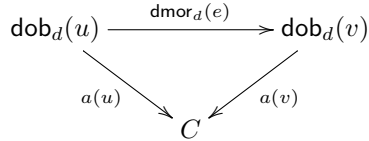
\[
 \begin{xy}
  \xymatrix{
                 \dob_d(u) \ar[rr]^{\dmor_d(e)} \ar[rd]_{a(u)}   &  &  \dob_d(v) \ar[ld]^{a(v)} \\
                             & C &
  }
 \end{xy}
\]

We often omit the equalities, denoting a cocone just by its family of
morphisms.

\begin{definition}[Colimiting cocone]
  A cocone $a$ under $d$ (of shape $G$) with tip $C$ is called
  \fat{colimiting} if for any cocone $a'$ under $d$ with tip $C'$
  there is exactly one morphism $f : \C(C,C')$ such that
  $\compose{a(v)}{f} = a'(v)$ for any $v : G$.
  Let $\iscolimitingCocone(d,C,a)$ denote this property.
\end{definition}

This definition can be illustrated by the following diagram:
  \[
 \begin{xy}
  \xymatrix@R=5pc{
                 \dob_d(u) \ar[rr]^{} \ar[rd]_{a(u)}   \ar@{-->}[rrrd]_>>>>>>>>>>>>>>>{a'(u)}|!{[dr];[rr]}\hole
                           &  &  \dob_d(v) \ar[ld]_>>>>>>>{a(v)} \ar@{-->}[rd]^{a'(v)}                           
                           \\
                             & C \ar@{.>}[rr]_{\exists ! f}& & C'
  }
 \end{xy}
\]

In \UniMath we represent this by:

\begin{coqdoccode}
\coqdocemptyline
\coqdocnoindent
\coqdockw{Definition} \coqdef{UniMath.CategoryTheory.limits.graphs.colimits.isColimCocone}{isColimCocone}{\coqdocdefinition{isColimCocone}} \{\coqdocvar{g} : \coqref{UniMath.CategoryTheory.limits.graphs.colimits.graph}{\coqdocdefinition{graph}}\} (\coqdocvar{d} : \coqref{UniMath.CategoryTheory.limits.graphs.colimits.diagram}{\coqdocdefinition{diagram}} \coqdocvariable{g} \coqdocvariable{C}) (\coqdocvar{c} : \coqdocvariable{C})\coqdoceol
\coqdocindent{1.00em}
(\coqdocvar{a} : \coqref{UniMath.CategoryTheory.limits.graphs.colimits.cocone}{\coqdocdefinition{cocone}} \coqdocvariable{d} \coqdocvariable{c}) : \coqref{UniMath.Foundations.Basics.Preamble.UU}{\coqdocdefinition{UU}} := \coqref{UniMath.Foundations.Basics.Preamble.:type scope:'xCExA0' x '..' x ',' x}{\coqdocnotation{\ensuremath{\Pi}}} \coqref{UniMath.Foundations.Basics.Preamble.:type scope:'xCExA0' x '..' x ',' x}{\coqdocnotation{(}}\coqdocvar{c'} : \coqdocvariable{C}) (\coqdocvar{a'} : \coqref{UniMath.CategoryTheory.limits.graphs.colimits.cocone}{\coqdocdefinition{cocone}} \coqdocvariable{d} \coqdocvariable{c'}\coqref{UniMath.Foundations.Basics.Preamble.:type scope:'xCExA0' x '..' x ',' x}{\coqdocnotation{),}}\coqdoceol
\coqdocindent{2.00em}
\coqref{UniMath.Foundations.Basics.PartA.iscontr}{\coqdocdefinition{iscontr}} (\coqref{UniMath.Foundations.Basics.Preamble.:type scope:'xCExA3' x '..' x ',' x}{\coqdocnotation{Σ}} \coqdocvar{x} : \coqdocvariable{C}\coqref{UniMath.CategoryTheory.UnicodeNotations.::x 'xE2x9FxA6' x ',' x 'xE2x9FxA7'}{\coqdocnotation{\leftdoublebracket}}\coqdocvariable{c}\coqref{UniMath.CategoryTheory.UnicodeNotations.::x 'xE2x9FxA6' x ',' x 'xE2x9FxA7'}{\coqdocnotation{,}}\coqdocvariable{c'}\coqref{UniMath.CategoryTheory.UnicodeNotations.::x 'xE2x9FxA6' x ',' x 'xE2x9FxA7'}{\coqdocnotation{\rightdoublebracket}}\coqref{UniMath.Foundations.Basics.Preamble.:type scope:'xCExA3' x '..' x ',' x}{\coqdocnotation{,}} \coqref{UniMath.Foundations.Basics.Preamble.:type scope:'xCExA0' x '..' x ',' x}{\coqdocnotation{\ensuremath{\Pi}}} \coqdocvar{v}\coqref{UniMath.Foundations.Basics.Preamble.:type scope:'xCExA0' x '..' x ',' x}{\coqdocnotation{,}} \coqref{UniMath.CategoryTheory.limits.graphs.colimits.coconeIn}{\coqdocdefinition{coconeIn}} \coqdocvariable{a} \coqdocvariable{v} \coqref{UniMath.CategoryTheory.UnicodeNotations.::x ';;' x}{\coqdocnotation{;;}} \coqdocvariable{x} \coqref{UniMath.Foundations.Basics.Preamble.:type scope:x '=' x}{\coqdocnotation{=}} \coqref{UniMath.CategoryTheory.limits.graphs.colimits.coconeIn}{\coqdocdefinition{coconeIn}} \coqdocvariable{a} \coqdocvariable{v}).\coqdoceol
\coqdocemptyline
\end{coqdoccode}

Here \coqdocdefinition{iscontr} is a predicate saying that the type is
contractible, in other words that it has only one inhabitant which
exactly captures the unique existence of $f$.

\begin{remark}[Uniqueness of colimits]

  If $\C$ is a \emph{univalent} category~\cite{rezk_completion}, and
  $d$ is a diagram of shape $G$ in $\C$, then the type of colimits of
  $d$,
  \[  \sum_{C:\C}\sum_{a : \Cocone(d,C)}\iscolimitingCocone(d,C,a) \enspace , \]
  is a proposition.
\end{remark}

Given a functor $F : \C\to\D$, a diagram $d$ in $\C$ and a cocone $a$
of $d$ with tip $C : \C_0$, then $Fa$ is a cocone under $Fd$ with tip
$FC$ in $\D$, where $Fa$ and $Fd$ are defined in the obvious way.
 
\begin{definition}[Preservation of colimits]
  Fix a graph $G$.
  We say that $F$
  \fat{preserves} 
  colimits of shape $G$ if, for any diagram $d$ of shape $G$ in $\C$,
  and any cocone $a$ under $d$ with tip $C$, 
  the cocone $Fa$ is colimiting for $Fd$ whenever $a$ is colimiting for $d$.
\end{definition}

A functor is called \fat{cocontinuous} if it preserves all colimits. In \UniMath:

\begin{coqdoccode}
\coqdocemptyline
\coqdocnoindent
\coqdockw{Definition} \coqdef{UniMath.CategoryTheory.CocontFunctors.preserves colimit}{preserves\_colimit}{\coqdocdefinition{preserves\_colimit}} \{\coqdocvar{g} : \coqref{UniMath.CategoryTheory.limits.graphs.colimits.graph}{\coqdocdefinition{graph}}\} (\coqdocvar{d} : \coqref{UniMath.CategoryTheory.limits.graphs.colimits.diagram}{\coqdocdefinition{diagram}} \coqdocvariable{g} \coqdocvariable{C}) (\coqdocvar{L} : \coqdocvariable{C})\coqdoceol
\coqdocindent{1.00em}
(\coqdocvar{cc} : \coqref{UniMath.CategoryTheory.limits.graphs.colimits.cocone}{\coqdocdefinition{cocone}} \coqdocvariable{d} \coqdocvariable{L}) : \coqref{UniMath.Foundations.Basics.Preamble.UU}{\coqdocdefinition{UU}} :=\coqdoceol
\coqdocindent{1.00em}
\coqref{UniMath.CategoryTheory.limits.graphs.colimits.isColimCocone}{\coqdocdefinition{isColimCocone}} \coqdocvariable{d} \coqdocvariable{L} \coqdocvariable{cc} \coqexternalref{:type scope:x '->' x}{http://coq.inria.fr/distrib/8.5pl3/stdlib/Coq.Init.Logic}{\coqdocnotation{\ensuremath{\rightarrow}}} \coqref{UniMath.CategoryTheory.limits.graphs.colimits.isColimCocone}{\coqdocdefinition{isColimCocone}} (\coqref{UniMath.CategoryTheory.CocontFunctors.mapdiagram}{\coqdocdefinition{mapdiagram}} \coqdocvariable{d}) (\coqdocvariable{F} \coqdocvariable{L}) (\coqref{UniMath.CategoryTheory.CocontFunctors.mapcocone}{\coqdocdefinition{mapcocone}} \coqdocvariable{d} \coqdocvariable{cc}).\coqdoceol
\coqdocemptyline
\coqdocnoindent
\coqdockw{Definition} \coqdef{UniMath.CategoryTheory.CocontFunctors.is cocont}{is\_cocont}{\coqdocdefinition{is\_cocont}} := \coqref{UniMath.Foundations.Basics.Preamble.:type scope:'xCExA0' x '..' x ',' x}{\coqdocnotation{\ensuremath{\Pi}}} \coqref{UniMath.Foundations.Basics.Preamble.:type scope:'xCExA0' x '..' x ',' x}{\coqdocnotation{\{}}\coqdocvar{g} : \coqref{UniMath.CategoryTheory.limits.graphs.colimits.graph}{\coqdocdefinition{graph}}\} (\coqdocvar{d} : \coqref{UniMath.CategoryTheory.limits.graphs.colimits.diagram}{\coqdocdefinition{diagram}} \coqdocvariable{g} \coqdocvariable{C}) (\coqdocvar{L} : \coqdocvariable{C})\coqdoceol
\coqdocindent{1.00em}
(\coqdocvar{cc} : \coqref{UniMath.CategoryTheory.limits.graphs.colimits.cocone}{\coqdocdefinition{cocone}} \coqdocvariable{d} \coqdocvariable{L}\coqref{UniMath.Foundations.Basics.Preamble.:type scope:'xCExA0' x '..' x ',' x}{\coqdocnotation{),}} \coqref{UniMath.CategoryTheory.CocontFunctors.preserves colimit}{\coqdocdefinition{preserves\_colimit}} \coqdocvariable{d} \coqdocvariable{L} \coqdocvariable{cc}.\coqdoceol
\coqdocemptyline
\end{coqdoccode}

A functor is called \fat{$\omega$-cocontinuous} if it preserves
colimits of diagrams of the shape
\[ 
  \begin{xy}
            \xymatrix@C=4pc{
	      A_0 \ar[r]^{f_0} & A_1  \ar[r]^{f_1} & A_2  \ar[r]^{f_2} & \ldots
            }
  \end{xy}
\]
that is, diagrams on the graph where objects are natural numbers and
where there is a unique arrow from $m$ to $n$ if and only if $1 + m =
n$. We refer to diagrams of this shape as \fat{chains}. 

Actually, in the formalization, the type of arrows from $m$ to $n$ is
defined to be the type of proofs that $1 + m = n$, exploiting the fact
that the type of natural numbers is a set:

\begin{coqdoccode}
\coqdocemptyline
\coqdocnoindent
\coqdockw{Definition} \coqdef{UniMath.CategoryTheory.CocontFunctors.nat graph}{nat\_graph}{\coqdocdefinition{nat\_graph}} : \coqref{UniMath.CategoryTheory.limits.graphs.colimits.graph}{\coqdocdefinition{graph}} :=
\coqref{UniMath.CategoryTheory.limits.graphs.colimits.mk graph}{\coqdocdefinition{mk\_graph}} \coqexternalref{nat}{http://coq.inria.fr/distrib/8.5pl3/stdlib/Coq.Init.Datatypes}{\coqdocinductive{nat}} (\coqref{UniMath.Foundations.Basics.Preamble.::'xCExBB' x '..' x ',' x}{\coqdocnotation{\ensuremath{\lambda}}} \coqdocvar{m} \coqdocvar{n}\coqref{UniMath.Foundations.Basics.Preamble.::'xCExBB' x '..' x ',' x}{\coqdocnotation{,}} 1 \coqexternalref{:nat scope:x '+' x}{http://coq.inria.fr/distrib/8.5pl3/stdlib/Coq.Init.Peano}{\coqdocnotation{+}} \coqdocvariable{m} \coqref{UniMath.Foundations.Basics.Preamble.:type scope:x '=' x}{\coqdocnotation{=}} \coqdocvariable{n}).\coqdoceol
\coqdocemptyline
\coqdocnoindent
\coqdockw{Notation} "'chain'" := (\coqref{UniMath.CategoryTheory.limits.graphs.colimits.diagram}{\coqdocdefinition{diagram}} \coqref{UniMath.CategoryTheory.CocontFunctors.nat graph}{\coqdocdefinition{nat\_graph}}).\coqdoceol
\coqdocemptyline
\coqdocnoindent
\coqdockw{Definition} \coqdef{UniMath.CategoryTheory.CocontFunctors.is omega cocont}{is\_omega\_cocont}{\coqdocdefinition{is\_omega\_cocont}} \{\coqdocvar{C} \coqdocvar{D} : \coqref{UniMath.CategoryTheory.precategories.precategory}{\coqdocdefinition{precategory}}\} (\coqdocvar{F} : \coqref{UniMath.CategoryTheory.functor categories.functor}{\coqdocdefinition{functor}} \coqdocvariable{C} \coqdocvariable{D}) : \coqref{UniMath.Foundations.Basics.Preamble.UU}{\coqdocdefinition{UU}} :=
\coqdoceol
\coqdocindent{1.00em}
\coqref{UniMath.Foundations.Basics.Preamble.:type scope:'xCExA0' x '..' x ',' x}{\coqdocnotation{\ensuremath{\Pi}}} \coqref{UniMath.Foundations.Basics.Preamble.:type scope:'xCExA0' x '..' x ',' x}{\coqdocnotation{(}}\coqdocvar{c} : \coqref{UniMath.CategoryTheory.CocontFunctors.omega cocont.::'chain'}{\coqdocnotation{chain}} \coqdocvariable{C}) (\coqdocvar{L} : \coqdocvariable{C}) (\coqdocvar{cc} : \coqref{UniMath.CategoryTheory.limits.graphs.colimits.cocone}{\coqdocdefinition{cocone}} \coqdocvariable{c} \coqdocvariable{L}\coqref{UniMath.Foundations.Basics.Preamble.:type scope:'xCExA0' x '..' x ',' x}{\coqdocnotation{),}}
\coqref{UniMath.CategoryTheory.CocontFunctors.preserves colimit}{\coqdocdefinition{preserves\_colimit}} \coqdocvariable{F} \coqdocvariable{c} \coqdocvariable{L} \coqdocvariable{cc}.\coqdoceol
\coqdocemptyline
\coqdocnoindent
\coqdockw{Definition} \coqdef{UniMath.CategoryTheory.CocontFunctors.omega cocont functor}{omega\_cocont\_functor}{\coqdocdefinition{omega\_cocont\_functor}} (\coqdocvar{C} \coqdocvar{D} : \coqref{UniMath.CategoryTheory.precategories.precategory}{\coqdocdefinition{precategory}}) : \coqref{UniMath.Foundations.Basics.Preamble.UU}{\coqdocdefinition{UU}} :=\coqdoceol
\coqdocindent{1.00em}
\coqref{UniMath.Foundations.Basics.Preamble.:type scope:'xCExA3' x '..' x ',' x}{\coqdocnotation{Σ}} \coqref{UniMath.Foundations.Basics.Preamble.:type scope:'xCExA3' x '..' x ',' x}{\coqdocnotation{(}}\coqdocvar{F} : \coqref{UniMath.CategoryTheory.functor categories.functor}{\coqdocdefinition{functor}} \coqdocvariable{C} \coqdocvariable{D}\coqref{UniMath.Foundations.Basics.Preamble.:type scope:'xCExA3' x '..' x ',' x}{\coqdocnotation{),}} \coqref{UniMath.CategoryTheory.CocontFunctors.is omega cocont}{\coqdocdefinition{is\_omega\_cocont}} \coqdocvariable{F}.\coqdoceol
\coqdocemptyline
\end{coqdoccode}

\begin{lemma}[Invariance of cocontinuity under isomorphism]
  Let $F,G : \C\to\D$ be functors, and let $\alpha : F \cong G$ be a
  natural isomorphism, then $G$ preserves colimits (of a certain
  shape) if $F$ does.
\end{lemma}

Note that, as preservation of colimits is a proposition, it suffices
for the natural isomorphism $\alpha$ to merely exist for the lemma to
hold.

Next, we construct colimits in the functor category from colimits in the target
category:
\begin{problem}[Colimits in functor categories]\label{problem:colims_in_functor_cat}
  Let $\C$ be a category, and let $\D$ be a category 
  with all (specified) colimits of a given shape.
  To construct colimits of the same shape in the functor category $[\C,\D]$.
\end{problem}

\begin{construction}[Solution to Problem~\ref{problem:colims_in_functor_cat}]\label{constr:colims_in_functor_cat}
 The construction of colimits in a functor category is pointwise:
 the colimit $C$ of a diagram is given, at point $c : \C_0$, 
 as the colimit in $\D$ of the diagram obtained by evaluating the diagram in $c : \C_0$.
\end{construction}

Limits have been formalized in the same way as colimits, that is,
parametrized by graphs and diagrams. We have implemented a similar construction
for lifting limits to functor categories. We omit the details of the
dualization.

  In the formalization some (co)limits (e.\,g., pullbacks and pushouts) 
  are also 
  implemented \emph{directly}, in addition to them being formalized as
  a colimit over a specific graph.  For instance, binary
  coproducts are formalized as a type parametrized by two objects in a
  category, instead of by a diagram on the graph $\mathbf{2}$ with two
  objects and no non-trivial morphisms.  We provide suitable maps going back and
  forth between the different implementations of (co)limits.
  
  The advantage of formalizing the `special' (co)limits as instances of 
  general (co)limits is that results such as the lifting of 
  (co)limits to functor categories restricts immediately to 
  these (co)limits of special shapes.
  
  On the other hand, the direct formulation is more convenient to 
  work with in practice. In particular, we experienced some performance
  issues in the compilation of our library when we attempted to replace
  the direct lifting of binary (co)products to functor categories by
  a specialization of the general lifting of (co)limits.
  Those performance issues are related to a `structure vs.\ property'
  question: the lifting should happen in such a way that the 
  binary product of two functors $F,G : \C \to \D$, 
  evaluated in an object $C : \C_0$, \emph{computes} 
  (that is, is judgmentally equal) to the binary product of $FC$ and $GC$,
  the latter of which was given by hypothesis as a structure.

\subsection{Initial algebras from colimits of chains}

The construction of initial algebras as colimits of chains was first
described by Adámek in~\cite{Adamek74}. It is a purely categorical
construction and the formalization presented no surprises.

\begin{problem}[Initial algebras of $\omega$-cocontinuous
  functors]\label{problem:init_alg}
  Let $\C$ be a category with initial object $0$, and let $F : \C\to\C$ be $\omega$-cocontinuous. 
  Let $c$ be a colimiting cocone with tip $C$
  of the chain $\chain_F$ given as follows:
  \[
  \begin{xy}
            \xymatrix@C=4pc{
	      0 \ar[r]^{!} \ar[rd]_{c_0} 
	         & F0  \ar[r]^{F!} \ar[d]_{c_1} 
	          & F^20  \ar[r]^{F^2!} \ar[ld]^{c_2} 
	           & \ldots   
\\ 	                   &  C & 
            }
  \end{xy}
  \]
  Equip $C$ with an $F$-algebra structure $\alpha : \C(FC,C)$ and show that
  $(C,\alpha)$ is an initial $F$\-algebra.
\end{problem}

To motivate the solution to the problem, we recall Lambek's well-known lemma that we also formalized.
\begin{lemma}[Lambek]
  Given $F : \C\to\C$ and an initial algebra $(A,a)$ of $F$, then
  $a : \C(FA,A)$ is an isomorphism.
\end{lemma}
\begin{proof}
The inverse arrow to $a$ is obtained as the unique algebra morphism 
to the algebra $(FA,Fa)$.
\end{proof}

Thanks to this result, we are bound to find an $\alpha : \C(FC,C)$ above that is even an isomorphism.
  
\begin{construction}[Solving Problem~\ref{problem:init_alg}]\label{thm:init_alg}
  In order to construct an isomorphism $\alpha : FC \cong C$, we use
  that we obtain an isomorphism between any two objects that are
  colimits for the same diagram.  It hence suffices to show that $FC$
  is the tip of a colimit of the above diagram. But $FC$ is a colimit of the
  diagram $F\chain_F$ obtained by applying $F$ to each object and arrow of
  $\chain_F$, by $\omega$-cocontinuity of $F$.  At the same time, the colimit
  of $F\chain_F$ is the same as of $\chain_F$, since the colimit of a chain remains
  the same under the \enquote{shift} of a chain, or, more generally, under
  the removal of a finite prefix of a chain (this is due to the fact that
  the cocones can always be ``completed leftwards'' by pure calculation).
  
  Given an algebra $(A,a)$, we have to construct a cocone under $\chain_F$
  with tip $A$ in order to obtain a morphism from $C$ to $A$. The
  cocone is defined by induction on natural numbers: the morphism of
  index $0$ is the one from the initial object. The morphism at index
  $n+1$ is constructed by composing $a$ with the image of that at
  index $n$ under $F$.  This forms a cocone, which induces a morphism
  $f : \C(C,A)$.  This morphism is also a morphism
  of algebras from $(C,\alpha)$ to $(A,a)$.  Its uniqueness is a
  consequence of it being unique as a morphism out of the tip $C$ of the colimit.
\end{construction}

\subsection{Colimits in $\set$}\label{sec:colimits_in_sets}

The construction of colimits in the category of sets we 
present in this section requires two consequences of the 
univalence axiom: function extensionality and univalence for propositions.

It is well-known that the construction of colimits can be split into
the construction of coproducts and the construction of coequalizers
(see \cite[p.\,113]{maclane} for the dual situation with limits).
Using this point of view, it is
the construction of coequalizers that is not possible 
in pure
Martin-Löf type theory (see, e.g., \cite{DBLP:conf/types/ChicliPS02})
and requires the aforementioned consequences of the univalence axiom.

\subsubsection{Set quotients in \UniMath}
\label{sec:quotients_in_unimath}

Set-level quotients were constructed by Voevodsky in his Foundations
library (which is now a part of \UniMath); a brief overview can be found
in~\cite{voevodsky_experimental}.  None of the work described in this
section is our own.

Given a type $X$, we call $\eqrel(X)$ the type of equivalence
relations $R : X\to\X\to\prop$, that is, 
reflexive, symmetric, and transitive relations.
For such an equivalence relation $R$, the \fat{set quotient} $X/R$,
together with the canonical surjection $\pr:X\to X/R$,
has the following universal property: for any set $S$
and map $f: X \to S$ such that $R(x,y)$ implies $f(x) = f(y)$, there
is a unique map $\hat{f} : X/R \to S$ such that the following diagram
commutes.
\[
\begin{xy}
  \xymatrix{
            X  \ar[rd]^{f} \ar[d]_{\pr} &  \\
            X/R \ar[r]_{\hat{f}} &  S                        
  }
\end{xy}
\]
Note that, for any $x, y : X$, we have
\begin{equation}
R(x,y) \simeq (\pr(x) = \pr(y)) \label{eq:equiv_quot_eq}
\end{equation}

The construction Voevodsky gives of the set quotient $X/R$
in terms of equivalence classes of $R$,
uses function extensionality and univalence for
propositions.

\subsubsection{Construction of colimits in $\set$}\label{smallsec:colimits_in_sets}

The goal of this section is a solution to the following problem:
\begin{problem}[Colimits in $\set$]\label{problem:colimits-in-sets}
 Given a graph $G$ and a diagram $d$ of shape $G$ in $\set$,
 to construct the colimit of $d$.
\end{problem}

\begin{construction}[Solution to Problem~\ref{problem:colimits-in-sets}]
\label{constr:colimits-in-sets}
The tip of the colimit of $d$ is given by
  \[ C:= \left(\sm{v:G}\dob_d(v)\right)/ \sim \]
  with $\sim$ being the
  smallest equivalence relation containing the relation $\sim_0$,
  defined by
  \[(u,A) \sim_0 (v,B) \text{ iff } \exists e \in\edge_G(u,v) \text{ with } \dmor_d(e)(A) = B \enspace . \]
  The colimiting cocone under $C$ is given by composing the projection
  $\pr$ with the injection maps
  $\C\bigl(d(u),\sm{v:G}\dob_d(v)\bigr)$. The fact that the family of
  maps thus obtained constitutes a cocone makes use of the  
  equivalence of \eqref{eq:equiv_quot_eq}. The (unique) map to
  any cocone is obtained by the universal property of the set quotient.
  Showing uniqueness of that map makes use of
  the fact that the projection is surjective, and hence an epimorphism
  in the category of sets.
\end{construction}

Note that in the above formula, we use the truncated existential
$\exists$ instead of the proof-relevant $\Sigma$. This is necessary in
order to give $\sim_0$ the target type $\prop$, and hence to apply the
construction of quotients described in
Section~\ref{sec:quotients_in_unimath}.

Note also that for the above construction to be correct, we need 
the type of vertices of $G$ to be small.
In the present work, we are ultimately interested in colimits of chains,
that is, of diagrams where the set of vertices is given by the set of 
natural numbers---a small set.

In order to construct the smallest equivalence relation containing a relation
  $R_0$, we need to close $R_0$ under
reflexivity, symmetry and transitivity:

\begin{definition}\label{def:eq_closure}
  Let $R_0 : X \to X \to \prop$ be a relation on a type $X$. Its
  closure is defined to be the relation $x\sim y$ given by
  \[ x \sim y := \prd{R : \eqrel(X)} (R_0 \subseteq R) \to R(x,y)  \]
\end{definition}
Here, we denote by $R_0 \subseteq R$ that $R_0(x,y)$ implies $R(x,y)$
for any $x,y : X$.
Note that this definition requires impredicativity for h-levels:
the fact that $x\sim y : \prop$ is a consequence of $R(x,y)$
being a proposition for any equivalence relation $R$.
We do not worry about the universe level of the 
relation $\sim$.

\begin{lemma}
  The relation defined in Definition~\ref{def:eq_closure} is the smallest
  equivalence relation containing $R_0$.
\end{lemma}
\begin{proof}
  Minimality is direct by the impredicative definition; $\sim$ is
  itself an equivalence relation because equivalence relations are
  closed under arbitrary intersections.
\end{proof}

\subsection{Functors preserving colimits}
\label{sec:inductive_sets}

In this section, we prove results on functors preserving colimits, in
particular colimits of chains. The first is a classical result about
preservation of colimits by left adjoints~\cite[p.~119]{maclane}.

\begin{lemma}
  \label{lemma:left_adjoint_preserves_colimits}
  If $F : \C \to \D$ is a left adjoint with right adjoint $G : \D \to
  \C$, then it preserves colimits.
\end{lemma}
\begin{proof}
  Call $\phi$ the (natural) family of isomorphisms $\phi_{C,D} : \D(FC,D) \simeq
  \C(C,GD)$ of the adjunction. We omit the subscripts in what follows.
  Given a colimiting cocone $(a_i)_{i : I}$ with tip $L$ for some diagram $d$, we
  need to show that the right-hand cocone is colimiting for the
  diagram $Fd$.
  \[
     \begin{xy}
       \xymatrix{
        A_i \ar[rr]^{f} \ar[rd]_{a_i} & & A_j \ar[ld]^{a_j} \\
             &L&
       }
     \end{xy}
 \hspace{1cm}
 \begin{xy}
       \xymatrix{
        FA_i \ar[rr]^{Ff} \ar[rd]_{Fa_i} & & FA_j \ar[ld]^{Fa_j} \\
             &FL&
       }
     \end{xy}
  \]
   We hence need to show that, for
  any cocone $(e_i)_i$ under $Fd$ with tip $M$, the type $
  \sum_{x:\D(FL, M)} \prd{i:I}\compose{Fa_i}{x}=e_i$ is contractible. We
  show that it is equivalent to a contractible one, and hence
  contractible itself:
  \begin{align*}
   \sum_{x:\D(FL, M)} \prd{i:I}\compose{Fa_i}{x}=e_i \enspace
   &\simeq \enspace \sum_{y:\C(L, GM)} \prd{i:I}\compose{Fa_i}{\inv{\phi}(y)}=e_i \\
   &\simeq \enspace \sum_{y:\C(L, GM)} \prd{i:I} \inv{\phi}(\compose{a_i}{y})=e_i \\
   &\simeq \enspace \sum_{y:\C(L, GM)} \prd{i:I} \compose{a_i}{y}= \phi(e_i) \\
   &\simeq \enspace 1
  \end{align*}
  The last equivalence is given by hypothesis for the cocone
  $\bigl(\phi(e_i)\bigr)_i$ with tip $GM$:
  
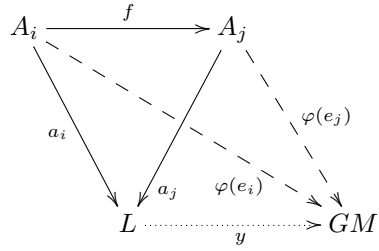
\[
     \begin{xy}
       \xymatrix@R=5pc{
        A_i \ar[rr]^{f} \ar[rd]_{a_i} \ar@{-->}[rrrd]_>>>>>>>>>>{\phi(e_i)}|!{[dr];[rr]}\hole & & A_j \ar[ld]^>>>>{a_j} \ar@{-->}[rd]^{\phi(e_j)}\\
             &L \ar@{.>}[rr]_{y}& & GM
       }
     \end{xy}
 \]
\end{proof}

In what follows we write $C^I$ for the $I$-indexed product category of
a category $C$.

\begin{lemma}[Examples of preservation of colimits]\label{lem:ex-pres-colim}\hfill
  \begin{enumerate}
    \item\label{pres-colim-id}
      The identity functor preserves colimits.
    \item\label{pres-colim-const}
      Any constant functor $\constfunctor{d} : \C \to \D$ preserves colimits of chains.
    \item\label{pres-colim-delta}
      If $\C$ has specified products, the diagonal functor
      $\Delta : \C \to \C^I$ mapping an object $X$ to the constant $I$-indexed family $\langle X \rangle_{i : I}$
      preserves colimits.
    \item\label{pres-colim-plus} 
      If $\C$ has specified coproducts, the functor $\amalg : \C^I \to \C$,
      mapping $I$-indexed families of $X_i$ to their coproduct, preserves colimits.
    \end{enumerate}
  \end{lemma}
\begin{proof}
  The points \ref{pres-colim-id} and \ref{pres-colim-const} are direct.
  The other two points follow by
  Lemma~\ref{lemma:left_adjoint_preserves_colimits}. Indeed, under the
  assumptions specified in each case we have adjunctions:
  \[
  \amalg \dashv \Delta \dashv \Pi
  \]
  
  where $\Pi : \C^I \to \C$ is the functor that maps $I$-indexed
  families of $X_i$ to their product.
\end{proof}

Note that point \ref{pres-colim-const} is only stated for chains, that
is because it is in general not true that constant functors preserve
colimits.

The next results state that various functors preserve cocontinuity of
all kinds. By this, we mean that if the input functors preserve
colimits of shape $G$ for a graph $G$, then so does the output
functor, in particular, this yields preservation of
$\omega$-cocontinuity (which does not follow from preservation of
cocontinuity).
\begin{lemma}[Examples of preservation of cocontinuity]\label{lem:ex-pres-cocont}\hfill
  \begin{enumerate}
    \item\label{pres-colim-comp}
      The composition of two functors preserves colimits of a certain
      kind, if the input functors do.

    \item\label{pres-colim-family}
      Given a family of functors $F_i : \C \to \D$ indexed by $i : I$,
      where $I$ has decidable equality. If all the $F_i$ preserve
      colimits of a certain kind, then the functor
      $\langle F_i \rangle_{i : I}: \C^I \to \D^I$ preserves colimits
      of that kind.
    \end{enumerate}
  \end{lemma}
\begin{proof}
  The first point is direct. For~\ref{pres-colim-family} we sketch the
  binary case.

  We first prove that the projection functors preserves colimits. For
  the first projection, $\pi_1 : \C^2 \to \C$, we are by assumption
  given a colimiting cocone
  
  \[
   \begin{xy}
       \xymatrix{
        (A_i,B_i) \ar[rr]^{} \ar[rd]_{(a_i,b_i)} & & (A_j,B_j) \ar[ld]^{(a_j,b_j)} \\
             &(L,M)&
       }
     \end{xy}
  \]

  and need to show that the cocone $(a_i)$ with tip $L$ is colimiting. Given a
  cocone $(a'_i)$ with tip $X$ this can be illustrated constructing the
  map $f$ in:
  \[
 \begin{xy}
  \xymatrix@R=5pc{
                 A_i \ar[rr]^{} \ar[rd]_{a_i}   \ar@{-->}[rrrd]_>>>>>>>>>>>>>>>{a'_i}|!{[dr];[rr]}\hole
                           &  &  A_j \ar[ld]_>>>>>>>{a_j} \ar@{-->}[rd]^{a'_j}                           
                           \\
                             & L \ar@{.>}[rr]_{\exists ! f}& & X
  }
 \end{xy}
\]
  From the cocone $(a'_i)$ we can form 
  \[
   \begin{xy}
       \xymatrix{
        (A_i,B_i) \ar[rr]^{} \ar[rd]_{(a'_i,b_i)} & & (A_j,B_j) \ar[ld]^{(a'_j,b_j)} \\
             &(X,M)&
       }
     \end{xy}
  \]
 
  and by assumption obtain a unique map from $(L,M)$ to $(X,M)$. This
  then gives us the desired map $f : L \to X$.

  The proof that the second projection functor preserves colimits is
  analogous. For this we need to construct a cocone over $(L,X)$ from
  a cocone $(b'_i)$ with $b'_i : B_i \to X$ instead.  Decidable
  equality on $I$ is needed for the general case for proving that
  $\pi_i : \C^I \to \C$ preserves colimits. Indeed, we need to be able
  to decide equality on the indices to construct the cocone whose tip
  contains $X$ at index $i$.
  
  Using that the projection functors preserve colimits it is direct to
  show that $\langle F_1,F_2 \rangle : \C^2 \to \D^2$ preserves
  colimits of a certain kind if $F_1$ and $F_2$ do so. Given a
  colimiting cocone $(A_i,B_i)$ with tip $(L,M)$ we obtain the
  colimiting cocones $(A_i)$ with tip $L$ and $(B_i)$ with tip $M$ by
  the above proofs. As $F_1$ and $F_2$ preserve colimits we get that
  $(F_1L,F_2M)$ is the colimit of $(F_1A_i,F_2B_i)$.
\end{proof}

It was quite cumbersome to formalize the proof of
point~\ref{pres-colim-family} above as we needed to define cocones
where the type of the tips depends on the decidable equality of
$I$. The interested reader may consult the formalization for details.

Using what we have defined so far we can define the coproduct of an
$I$-indexed family of functors $F_i : \C \to \D$ by:
\[
\bigoplus_{i : I} F_i = \amalg \circ \langle F_i \rangle_{i : I} \circ \Delta
\]

On an object $X$ this functor acts by:
\[
  X \mapsto \langle X\rangle_{i : I}
    \mapsto \langle F_i X \rangle_{i : I}
    \mapsto \underset{i : I}{\coprod}~F_i X
\]

Being the composition of ($\omega$-)cocontinuous functors this is also
($\omega$-)cocontinuous.

We now turn our attention to the binary version of the product
functor, which we denote by $\times : \C^2 \rightarrow \C$. In order
to show that this functor is $\omega$-cocontinuous we need more
structure on the category $\C$.

\begin{definition}[Exponentials]\label{def:exponentials}
  Let $\C$ have specified binary products. 
  An exponential structure for $\C$ is, 
  for any $A : \C_0$, 
  a right adjoint for 
  the functor $A\times{-}$ given
  on objects by \( X \mapsto A \times X\).
  Given an exponential structure on $\C$, we 
  denote the right adjoint of $A\times{-}$
  by $({-})^A$. That is, on objects it acts as $B \mapsto B^A$.
\end{definition}
\begin{example}
  The exponential structure on the category $\Set$ is given, for the functor $A\times{-}$,
  by the functor given on object $B$ by  $B^A = A \to B$.
\end{example}
The functor ${-}\times A$ is defined analogously for each
$A:\C_0$. The functors $A \times{-}$ and ${-}\times A$ are naturally
isomorphic, so if one of them has a right adjoint the other does as
well. Hence the choice of which argument is fixed in
Definition~\ref{def:exponentials} is not crucial. The following
lemma is another instance of
Lemma~\ref{lemma:left_adjoint_preserves_colimits}:
\begin{lemma}
  Let $\C$ have (specified) binary products and exponentials, and let $A : \C_0$.  
  The functors $A \times{-}$ and ${-}\times A$ preserves colimits. 
\end{lemma}

Only the next result is specifically about $\omega$-cocontinuity.  A
search for existing proofs of this theorem in the literature only
revealed a sketch in an online resource \cite{metayer};
however, we have not found a precise proof of it.  Here, we give a
direct proof of this theorem.  While the proof idea is simple, writing
out all the details in the formalization is quite complicated.
Our outline here is not more detailed than the one in~\cite{metayer}, but
we have the advantage of being able to refer to the formalization for
details.

\begin{theorem}\label{theorem:times_cocont}
  Let $\C$ be a category with specified binary products such that $A\times{-}$ and ${-}\times B$ are
  $\omega$-cocontinuous for all $A,B:\C_0$. Then the functor $\times : \C^2 \to \C$ is
  $\omega$-cocontinuous.
\end{theorem}
\begin{proof}
    Given a diagram
  \[ 
    \begin{xy}
        \xymatrix@C=4pc{
          (A_0,B_0) \ar[r]^{(f_0,g_0)} & (A_1,B_1)  \ar[r]^{(f_1,g_1)} & (A_2,B_2)  \ar[r]^{(f_2,g_2)} & \ldots
        }
    \end{xy}
  \]
  
  with colimit $(L,R)$ (we omit the cocone maps), we need to show that
  $L\times R$ is the colimit of
  \[
    \begin{xy}
        \xymatrix@C=4pc{
           A_0 \times B_0 \ar[r]^{f_0 \times g_0} & A_1 \times B_1  \ar[r]^{f_1 \times g_1} & A_2 \times B_2  \ar[r]^{f_2 \times g_2} & \ldots
        }
    \end{xy}
  \]
  To this end, we consider the grid
  \[
    \begin{xy}
        \xymatrix@C=4pc{
            (A_0,B_0) \ar[r]^{(f_0,1)} \ar[d]_{(1,g_0)}& (A_1,B_0)  \ar[r]^{(f_1,1)} \ar[d]_{(1,g_0)}& (A_2,B_0)  \ar[r]^{(f_2,1)} \ar[d] & \ldots \\
	      (A_0,B_1) \ar[r]^{(f_0,1)} \ar[d]_{(1,g_1)} & (A_1,B_1)  \ar[r]^{(f_1,1)} \ar[d]& (A_2,B_1)  \ar[r]^{(f_2,1)} \ar[d]& \ldots \\
	         \vdots &  \vdots & \vdots & \vdots
         }
    \end{xy}
  \]
  The idea is to first take the colimit in each column, and then to
  take the colimit of the chain of colimits thus obtained.  In
  slightly more detail, by hypothesis, the colimit of the $i$th column
  is given by $A_i \times R$.  This gives rise to a chain $A_{i}
  \times R \to A_{i+1}\times R$, the limit of which is given by
  $L\times R$.  The difficult part of the proof is actually the
  handling of the arrows involved, something we completely omitted in
  this sketch.
\end{proof}
Note that if $\C$ has exponentials then the conditions of the lemma are
fulfilled. Hence it applies in particular to $\set$, or any other
cartesian closed category.

Using what we have defined so far, it is possible to construct many
datatypes, for example lists or binary trees over sets. 

\begin{example}[Lists of sets]\label{ex:lists}

Lists over a set $A$ can be defined as the initial algebra of the
following endofunctor on $\Set$ (using our notation for constant functors):
\[
L_A = \constfunctor 1 + \constfunctor A \times \Id{\Set}
\]

which, when evaluated at a set $X$, is $L_A(X) = 1 + A \times X$.
In \UniMath this is written as:

\begin{coqdoccode}
\coqdocemptyline
\coqdocnoindent
\coqdockw{Definition} \coqdef{UniMath.CategoryTheory.Inductives.Lists.L A}{L\_A}{\coqdocdefinition{L\_A}} : \coqref{UniMath.CategoryTheory.CocontFunctors.omega cocont functor}{\coqdocdefinition{omega\_cocont\_functor}} \coqref{UniMath.CategoryTheory.category hset.HSET}{\coqdocabbreviation{HSET}} \coqref{UniMath.CategoryTheory.category hset.HSET}{\coqdocabbreviation{HSET}} := \coqref{UniMath.CategoryTheory.CocontFunctors.:cocont functor hset scope:'''' x}{\coqdocnotation{'}}1 \coqref{UniMath.CategoryTheory.CocontFunctors.:cocont functor hset scope:x '+' x}{\coqdocnotation{+}} \coqref{UniMath.CategoryTheory.CocontFunctors.:cocont functor hset scope:'''' x}{\coqdocnotation{'}}\coqdocvariable{A} \coqref{UniMath.CategoryTheory.CocontFunctors.:cocont functor hset scope:x '*' x}{\coqdocnotation{\ensuremath{\times}}} \coqref{UniMath.CategoryTheory.CocontFunctors.:cocont functor hset scope:'Id'}{\coqdocnotation{Id}}.\coqdoceol
\coqdocemptyline
\end{coqdoccode}

Here \verb+HSET+ is the category $\Set$. This definition directly
produces an $\omega$-cocontinuous functor by exploiting the \Coq
notation mechanism and the packaging of functors with a proof that they
are $\omega$-cocontinuous.

By Construction~\ref{thm:init_alg} this has an initial algebra
consisting of $\mu L_A : \Set$ (representing lists of $A$) and a morphism
$\alpha : L_A(\mu L_A) \to \mu L_A$. 
If we expand the type of the morphism we
get
\[
\alpha : 1 + A \times \mu L_A \to \mu L_A
\]
and by precomposing with the injection maps into the coproduct we obtain:
\begin{align*}
\texttt{nil\_map} & : 1 \to \mu L_A\\
\texttt{cons\_map} & : A \times \mu L_A \to \mu L_A
\end{align*}
We write \texttt{nil} for \texttt{nil\_map tt} of type $\mu L_A$ (here
\texttt{tt} denotes the canonical element of the terminal set $1$)
and \texttt{cons} for the curried version of \texttt{cons\_map} whose
type is $A \to \mu L_A \to \mu L_A$.
As their names indicate, they correspond to the standard constructors
for lists where \texttt{nil} is the empty list and \texttt{cons} adds
an element to the front of a list.

Given a set $X$, an element $x : X$ and a function
$f : A \times X \to X$ we can construct another $L$-algebra by
$(X,[\lambda \_.x,f])$ where $[\lambda \_.x,f]$ is the coproduct of
the constant map to $x$ with $f$ and hence of type
$1 + A \times X \to X$. By initiality of $(\mu L_A,\alpha)$ we get an
$L$-algebra morphism $\texttt{foldr} : \mu L_A \to X$ satisfying:
\[
  \begin{xy}
   \xymatrix{
                  **[l]1 + A \times \mu L_A  \ar[r]^{\alpha} \ar[d]_{L_A(\texttt{foldr})} &   **[r]\mu L_A \ar[d]^{\texttt{foldr}}\\
                  **[l]1 + A \times X  \ar[r]^{[\lambda \_.x,f]}                 &   **[r]X \\
     }
  \end{xy}
\]

By precomposing with the injection maps this commutative diagram gives
us the equations:
\begin{align*}
\texttt{foldr nil} &= x\\
\texttt{foldr}~(\texttt{cons}~y~ys) &= f~(y,\texttt{foldr}~ys)
\end{align*}
These are the usual computation rules (modulo currying and implicit
arguments) of the \texttt{foldr} function as defined, for example, in
\Haskell. Hence this defines a recursion principle. We can also obtain
an induction principle:

\begin{coqdoccode}
\coqdocemptyline
\coqdocnoindent
\coqdockw{Lemma} \coqdef{UniMath.CategoryTheory.Inductives.Lists.listIndhProp}{listIndhProp}{\coqdoclemma{listIndhProp}} (\coqdocvar{P} : \coqref{UniMath.CategoryTheory.Inductives.Lists.List}{\coqdocdefinition{List}} \coqref{UniMath.Foundations.Basics.Preamble.:type scope:x 'xE2x86x92' x}{\coqdocnotation{→}} \coqref{UniMath.Foundations.Basics.Propositions.hProp}{\coqdocdefinition{hProp}}) :\coqdoceol
\coqdocindent{1.00em}
\coqdocvariable{P} \coqref{UniMath.CategoryTheory.Inductives.Lists.nil}{\coqdocdefinition{nil}} \coqref{UniMath.Foundations.Basics.Preamble.:type scope:x 'xE2x86x92' x}{\coqdocnotation{→}} \coqref{UniMath.Foundations.Basics.Preamble.:type scope:x 'xE2x86x92' x}{(}\coqref{UniMath.Foundations.Basics.Preamble.:type scope:'xCExA0' x '..' x ',' x}{\coqdocnotation{\ensuremath{\Pi}}} \coqdocvar{a} \coqdocvar{l}\coqref{UniMath.Foundations.Basics.Preamble.:type scope:'xCExA0' x '..' x ',' x}{\coqdocnotation{,}} \coqdocvariable{P} \coqdocvariable{l} \coqref{UniMath.Foundations.Basics.Preamble.:type scope:x 'xE2x86x92' x}{\coqdocnotation{→}} \coqdocvariable{P} (\coqref{UniMath.CategoryTheory.Inductives.Lists.cons}{\coqdocdefinition{cons}} \coqdocvariable{a} \coqdocvariable{l})\coqref{UniMath.Foundations.Basics.Preamble.:type scope:x 'xE2x86x92' x}{)} \coqref{UniMath.Foundations.Basics.Preamble.:type scope:x 'xE2x86x92' x}{\coqdocnotation{→}} \coqref{UniMath.Foundations.Basics.Preamble.:type scope:'xCExA0' x '..' x ',' x}{\coqdocnotation{\ensuremath{\Pi}}} \coqdocvar{l}\coqref{UniMath.Foundations.Basics.Preamble.:type scope:'xCExA0' x '..' x ',' x}{\coqdocnotation{,}} \coqdocvariable{P} \coqdocvariable{l}.\coqdoceol
\coqdocemptyline
\end{coqdoccode}

Using all of this we can define standard functions on these lists, for
example \texttt{map} and \texttt{length}, and prove some of their properties:

\begin{coqdoccode}
\coqdocemptyline
\coqdocnoindent
\coqdockw{Definition} \coqdef{UniMath.CategoryTheory.Inductives.Lists.length}{length}{\coqdocdefinition{length}} : \coqref{UniMath.CategoryTheory.Inductives.Lists.List}{\coqdocdefinition{List}} \coqexternalref{:type scope:x '->' x}{http://coq.inria.fr/distrib/8.5pl3/stdlib/Coq.Init.Logic}{\coqdocnotation{\ensuremath{\rightarrow}}} \coqexternalref{nat}{http://coq.inria.fr/distrib/8.5pl3/stdlib/Coq.Init.Datatypes}{\coqdocinductive{nat}} := \coqref{UniMath.CategoryTheory.Inductives.Lists.foldr}{\coqdocdefinition{foldr}} \coqref{UniMath.CategoryTheory.category hset.natHSET}{\coqdocdefinition{natHSET}} 0 (\coqref{UniMath.Foundations.Basics.Preamble.::'xCExBB' x '..' x ',' x}{\coqdocnotation{\ensuremath{\lambda}}} \coqdocvar{\_} (\coqdocvar{n} : \coqexternalref{nat}{http://coq.inria.fr/distrib/8.5pl3/stdlib/Coq.Init.Datatypes}{\coqdocinductive{nat}}\coqref{UniMath.Foundations.Basics.Preamble.::'xCExBB' x '..' x ',' x}{\coqdocnotation{),}} 1 \coqexternalref{:nat scope:x '+' x}{http://coq.inria.fr/distrib/8.5pl3/stdlib/Coq.Init.Peano}{\coqdocnotation{+}} \coqdocvariable{n}).\coqdoceol
\coqdocemptyline
\coqdocnoindent
\coqdockw{Definition} \coqdef{UniMath.CategoryTheory.Inductives.Lists.map}{map}{\coqdocdefinition{map}} (\coqdocvar{f} : \coqref{UniMath.CategoryTheory.Inductives.Lists.lists.::'A'}{\coqdocnotation{A}} \coqexternalref{:type scope:x '->' x}{http://coq.inria.fr/distrib/8.5pl3/stdlib/Coq.Init.Logic}{\coqdocnotation{\ensuremath{\rightarrow}}} \coqref{UniMath.CategoryTheory.Inductives.Lists.lists.::'A'}{\coqdocnotation{A}}) : \coqref{UniMath.CategoryTheory.Inductives.Lists.List}{\coqdocdefinition{List}} \coqexternalref{:type scope:x '->' x}{http://coq.inria.fr/distrib/8.5pl3/stdlib/Coq.Init.Logic}{\coqdocnotation{\ensuremath{\rightarrow}}} \coqref{UniMath.CategoryTheory.Inductives.Lists.List}{\coqdocdefinition{List}} :=\coqdoceol
\coqdocindent{1.00em}
\coqref{UniMath.CategoryTheory.Inductives.Lists.foldr}{\coqdocdefinition{foldr}} \coqdocvar{\_} \coqref{UniMath.CategoryTheory.Inductives.Lists.nil}{\coqdocdefinition{nil}} (\coqref{UniMath.Foundations.Basics.Preamble.::'xCExBB' x '..' x ',' x}{\coqdocnotation{\ensuremath{\lambda}}} \coqref{UniMath.Foundations.Basics.Preamble.::'xCExBB' x '..' x ',' x}{\coqdocnotation{(}}\coqdocvar{x} : \coqref{UniMath.CategoryTheory.Inductives.Lists.lists.::'A'}{\coqdocnotation{A}}) (\coqdocvar{xs} : \coqref{UniMath.CategoryTheory.Inductives.Lists.List}{\coqdocdefinition{List}}\coqref{UniMath.Foundations.Basics.Preamble.::'xCExBB' x '..' x ',' x}{\coqdocnotation{),}} \coqref{UniMath.CategoryTheory.Inductives.Lists.cons}{\coqdocdefinition{cons}} (\coqdocvariable{f} \coqdocvariable{x}) \coqdocvariable{xs}).\coqdoceol
\coqdocemptyline
\coqdocnoindent
\coqdockw{Lemma} \coqdef{UniMath.CategoryTheory.Inductives.Lists.length map}{length\_map}{\coqdoclemma{length\_map}} (\coqdocvar{f} : \coqref{UniMath.CategoryTheory.Inductives.Lists.lists.::'A'}{\coqdocnotation{A}} \coqexternalref{:type scope:x '->' x}{http://coq.inria.fr/distrib/8.5pl3/stdlib/Coq.Init.Logic}{\coqdocnotation{\ensuremath{\rightarrow}}} \coqref{UniMath.CategoryTheory.Inductives.Lists.lists.::'A'}{\coqdocnotation{A}}) : \coqref{UniMath.Foundations.Basics.Preamble.:type scope:'xCExA0' x '..' x ',' x}{\coqdocnotation{\ensuremath{\Pi}}} \coqdocvar{xs}\coqref{UniMath.Foundations.Basics.Preamble.:type scope:'xCExA0' x '..' x ',' x}{\coqdocnotation{,}} \coqref{UniMath.CategoryTheory.Inductives.Lists.length}{\coqdocdefinition{length}} (\coqref{UniMath.CategoryTheory.Inductives.Lists.map}{\coqdocdefinition{map}} \coqdocvariable{f} \coqdocvariable{xs}) \coqref{UniMath.Foundations.Basics.Preamble.:type scope:x '=' x}{\coqdocnotation{=}} \coqref{UniMath.CategoryTheory.Inductives.Lists.length}{\coqdocdefinition{length}} \coqdocvariable{xs}.\coqdoceol
\coqdocemptyline
\end{coqdoccode}

Note that the \coqdocdefinition{foldr} function in the formalization
takes a curried function as opposed to the one above. 

The computation rules for these lists do not hold definitionally, this
make them a little cumbersome to work with as one has to rewrite with
the equations above explicitly instead of letting \Coq do the
simplifications automatically. This is discussed further in
Section~\ref{sec:conclusion}.
\end{example}

We have also defined binary trees analogously to lists as the initial
algebra of the functor that maps $X$ to $1 + A \times X \times X$. It is hence
possible to introduce various homogeneous
datatypes using what has been developed so far.

  For nested datatypes, such as the introductory example of lambda
  terms, we can just try to use $[\C,\C]$ instead of the base category
  $\C$. While this is the right solution in principle, there are some
  technical details to be addressed to make this work. This is done in
  the next section which allows us to define heterogeneous nested
  datatypes representing syntax of languages with binders.

\subsection{The datatype specified by a binding signature}\label{sec:inductive_families}

In the introduction, we showed the motivating code example of a
representation of lambda terms by the family \verb+LC+ of types that
we qualified as \emph{nested datatype}, a name due to \cite{birdmeertens}.
In general, nested datatypes are datatypes that consist of a family of
types that are indexed over all types and where the constructors of
the datatype relate different family members. The homogeneous lists
are indexed over all types, but are no nested datatype since each
\verb+list X+ can be understood individually, while \verb+LC+ has the
constructor \verb+Abs+ that relates representations of lambda terms
with different sets of free variables. Being indexed ``over all
types'' needs to be specified properly. For us, it means that the
indexing parameter of the family runs through the objects of the same
category $\C$ that serves to represent the family members. In
particular, there is no inductive definition of a suitable maximal
indexing set, such as the natural numbers to represent a countably
infinite supply of ``fresh'' variable names.

From the point of view of category theory, nested datatypes are endofunctors
on a category $\C$ that arise as fixed points (up to isomorphism)
of endofunctors on $[\C,\C]$.
In the present work, we exclusively study fixed points given by initial
algebras. We do not insist on the datatype to be truly a nested
datatype in the above sense of relating different family members through the constructors.
Nonetheless, we want to capture the general situation where
indices of family members in the arguments of datatype constructors
are calculated by an arbitrary functor $F$. As illustrated in Example~\ref{ex:sig_strength_lc},
this calculation is done by using precomposition with that functor, in the example with
$F=\option$ that represents \enquote{context extension}.
Indeed, looking at the example, we see that variable binding is indicated by a
summand in the signature functor that maps an endofunctor $X$ to $X \hcomp \option$.

So, in order to construct nested datatypes in our setting, we would
like to show that functors on functor categories of the form
$\_\hcomp F : [\B,\C] \to [\A,\C]$
(with $F:[\A,\B]$) are $\omega$-cocontinuous, i.e., preserve colimits of chains.
Ultimately, we are interested in the case where $\A = \B = \C$,
but we prove a more general theorem below.

First, we need some auxiliary results.

\begin{lemma}\label{lem:colim-invariant-iso}
  Let $G$ be a graph and $D$ be a diagram of shape $G$ in $\C$. Given two cocones 
  with tips $C$ and $C'$, respectively, such that 
  the cocone with tip $C$ is colimiting, then the cocone with tip $C'$
  is colimiting if and only if the induced morphism from $C$ to $C'$ is an isomorphism.
\end{lemma}

\begin{theorem}\label{thm:colimits_pointwise}
  Fix a graph $G$ and assume $\C$ has colimits of shape $G$.
  Given a diagram $D$ of shape $G$ in the functor category $[\A,\C]$
  and a cocone with tip $F$,
  then this cocone is colimiting if and only if for any object $A : \A_0$
  the ``pointwise'' cocone with tip $FA$ is colimiting for the pointwise diagram $DA$ in $\C$.
\end{theorem}

\begin{proof}

In the proof we only mention the tip $F$ of the cocone, but
formally we have to handle the whole cocone.
 
 First, suppose that $F$ is a colimit.
 For any $A : \A_0$, we have the colimit, say $F'A$, of $DA$ in $\C$.
 Via Construction~\ref{constr:colims_in_functor_cat}, the pointwise colimits $F'A$ yield a functor $F'$ that is a colimit of $D$.
 Since both $F$ and $F'$ are colimits of $D$, we obtain an 
 isomorphism $F' \cong F$ by Lemma~\ref{lem:colim-invariant-iso}, and hence an isomorphism $FA \cong F'A$ for any $A : \A_0$.
 Since $F'A$ is a colimit for $DA$, so is $FA$.

On the other hand, suppose that $FA$ is a colimit of $DA$ for any $A : \A_0$.
Lifting those colimits to the functor category, we obtain a functor $F'$, that is
definitionally equal to $F$ on objects, and that is a colimit of $D$.
The induced natural transformation from $F'$ to $F$ is an isomorphism $F \cong F'$ that is pointwise the identity.
By Lemma~\ref{lem:colim-invariant-iso}, since $F'$ is a colimit of $D$, so is $F$.

\end{proof}

Using this we can now prove the main technical contribution of this section.

\begin{theorem}[Precomposition functor preserves colimits]\label{thm:precomp_cocont}
Fix a graph $G$ and suppose $\C$ has specified colimits of shape $G$.
Let $F : \A \to \B$ be a functor, then the functor $\_ \hcomp F : [\B,\C] \to [\A,\C]$
 preserves colimits of shape $G$.

\end{theorem}

\begin{proof}

 Let $D$ be a diagram of shape $G$ in $[\B,\C]$, and let $C$ be its colimit.
 We need to show that $C\hcomp F$ is the colimit of the diagram $G\hcomp F$ in $[\A,\C]$.
 By Theorem~\ref{thm:colimits_pointwise}, it suffices to show that for any $A : \A_0$, the object $(C\hcomp F)A \equiv C(FA)$ is a colimit of
 $(G\hcomp F)A \equiv G(FA)$ in $\C$.
 By the other implication of Theorem~\ref{thm:colimits_pointwise}, instantiated to $FA$, this is indeed the case.

\end{proof}

\begin{example}
 Putting together results~\ref{thm:precomp_cocont},  \ref{lem:ex-pres-colim}\ref{pres-colim-plus} in the binary case,
\ref{lem:ex-pres-cocont}\ref{pres-colim-family}, and \ref{theorem:times_cocont}, we obtain 
 that the functor for the untyped lambda calculus of Example~\ref{ex:sig_strength_lc} defined on objects as
 \begin{align*}
     X &\mapsto  \langle X,X\rangle  \\
       &\mapsto  \langle X \hcomp \option , X \times X \rangle \\
       &\mapsto  X \hcomp \option + X \times X
 \end{align*}
 is $\omega$-cocontinuous, being the composition of $\omega$-cocontinuous functors.
 Hence initial algebras can be constructed for it by Construction~\ref{thm:init_alg}.
 Note that we have not taken into account the variables yet. 
 This will be done below.
\end{example}

More generally,
any signature functor over a category $\C$ 
obtained from a binding signature
via Construction~\ref{constr_H:sem_sig_from_binding} 
preserves colimits of chains:

\begin{theorem}\label{lem:sem_sig_cocont}
  Let $\C$ be a
  category with coproducts, products, and colimits of chains such that
  $F\times {-}$ is $\omega$-cocontinuous for every $F:\C\to\C$.
  Then,
  the signature functor over $\C$ associated to a binding signature via
  Construction~\ref{constr_H:sem_sig_from_binding} 
  is $\omega$-cocontinuous.
\end{theorem}
By Lemma~\ref{lemma:left_adjoint_preserves_colimits}, the last requirement on $\C$ is satisfied if $\C$ has exponentials,
thus the theorem applies to $\C=\Set$. We also remark that the theorem uses the lifting
of colimits to functor categories (Construction \ref{constr:colims_in_functor_cat}).

\medskip
The binding signatures studied in Section~\ref{sec:signature} are
incapable of expressing that the free variables in the language are
considered as legal expressions, as we will argue now.
Had we also $\var:I$ in
Example~\ref{ex:binding_sig_lc}, any element of
$\arity(\var)$ would mean a lambda-term as argument to the
constructor, and if $\arity(\var)$ were the empty list, this would
generate one constant only. On the level of signature functors,
however, we just have to replace the $H$ found by
Construction~\ref{constr_H:sem_sig_from_binding} by
$\constfunctor{\Id{\C}} + H$.
Indeed, for any $(\constfunctor{\Id{\C}}+H)$-algebra
$(T,\alpha)$, the natural transformation $\alpha : \Id{\C} + HT \to T$
decomposes
into two $[\C,\C]$-morphisms
$\et{} : \Id{C} \arr T$, $\ta{} : H T \arr T$ defined by
\[
\begin{array}{c}
\et{} = \al{} \vcomp \inl_{\Id{\C},HT} \qquad\mbox{and}\qquad
\ta{} = \al{} \vcomp \inr_{\Id{\C},HT} \enspace .
\end{array}
\]
In case $(T,\alpha)$ is an initial algebra, the first component $\eta$
can then be considered as the injection of variables into the
well-formed expressions, i.\,e., for every object $C:\C$, $\eta_C:C\to
TC$ injects $C$ as ``variable names'' into $TC$, the ``terms over
$C$''. The second component $\tau$ represents all the other constructors
of $T$ together, hence those specified by the binding signature we started with.

\begin{definition}\label{def:datatype_binding_sig}
 The datatype specified by a signature functor $H$ over $\C$ (and hence by a binding signature)
 is given by an initial algebra of $\constfunctor{\Id{\C}} + H$.
\end{definition}

Combining Theorem~\ref{lem:sem_sig_cocont} with 
Ad\'amek's Theorem (Construction~\ref{thm:init_alg}), we obtain

\begin{construction}[Datatypes specified by binding signatures]\label{thm:initial_alg_for_binding_sig}
           Let $\C$ be a category with coproducts,
           products, and colimits of chains such that $F\times{-}$
           is $\omega$-cocontinuous for every $F:\C\to\C$.  For any
           binding signature $(I,\arity)$, construct the
           $\omega$-cocontinuous signature functor $H$. Then,
           $\constfunctor{\Id{\C}} + H$ is
           $\omega$-cocontinuous. Construct the datatype over $\C$ as
           initial algebra of the latter functor, where we get the required 
           colimiting cocone of Construction~\ref{thm:init_alg} 
           from $\C$ having specified colimits of chains.
           In particular,
           denoting the carrier of the algebra by $T:\C\to\C$, this
           yields $\et{} : \Id{C} \arr T$, $\ta{} : H T \arr T$ such
           that $[\et{},\ta{}]$ is an isomorphism.
\end{construction}
Once again, for $\C=\Set$, the prerequisites of the construction are
met, in particular thanks to the construction of colimits in the
category of sets (Construction~\ref{constr:colimits-in-sets}).

\section{From binding signatures to monads}\label{sec:thewholechain}

In this section we combine the results of the previous sections with
the construction of a substitution operation on an initial algebra in
order to obtain a ``substitution'' monad from a binding signature. We
end the section with two examples: the untyped lambda calculus and a
variation of Martin-Löf type theory.

\subsection{A substitution operation on the datatype of a binding signature}\label{sec:hss}

The results of the previous section permit the construction of initial algebras of signature functors.
The purpose of this section is to construct a \fat{substitution} operation on such initial algebras.
To this end, we apply Theorem~\ref{thm:subst_on_init} (a variant of a theorem from previous work, stated below) to our specific situation.
The goal of this section is hence
to recall the previous results and discuss some necessary modifications.

Even if not only initial algebras are considered (e.\,g., one might aim
at inverses of final coalgebras to model coinductive syntax, as
was one of the motivations for \cite{DBLP:journals/tcs/MatthesU04}),
the following abstract definition of the existence of a substitution
operation makes sense.
\begin{definition}[Matthes and Uustalu \cite{DBLP:journals/tcs/MatthesU04}]\label{def:hss}
Given a signature with strength  $(H, \th{})$,
 we call an $(\constfunctor{\Id{\C}}+H)$-algebra $(T, \al{})$ a
\fat{heterogeneous substitution system} (or \enquote{hss} for short) for
$(H, \th{})$, if, for every $\Ptd(\C)$-morphism $f : (Z,e) \arr (T,
\et{})$, there exists a unique $[\C,\C]$-morphism $h : T \hcomp Z \arr
T$, denoted $\gst{(Z,e)}{f}$, satisfying
\[
\begin{array}{c}
\xymatrix@C2pc@R1pc{
**[l] Z + (H T) \hcomp Z \ar[d]_{1_Z + \th{T,(Z,e)}} \ar[r]^-{\al{} \hcomp Z}
  & T \hcomp Z \ar[dd]^{h} \\
**[l] Z + H (T \hcomp Z) \ar[d]_{1_Z + H h}
  & \\
**[l] Z + H T \ar[r]^-{\copair{f}{\ta{}}}
  & T
}
\hspace{2mm} \mbox{i.e.,} \hspace{2mm}
\xymatrix@C3pc@R1pc{
Z \ar[r]^{\et{} \hcomp Z} \ar[ddr]_{f}
    & T \hcomp Z \ar[dd]^{h}
        & (H T) \hcomp Z  \ar[l]_-{\ta{} \hcomp Z} \ar[d]^{\th{T, (Z,e)}} \\
    &
        & H (T \hcomp Z) \ar[d]^{H h} \\
    & T
        & H T  \ar[l]_-{\ta{}}
}
\end{array}
\] 
\end{definition}
We remark that $(T, \al{})$ being an hss for given $(H, \th{})$ is a
proposition. Nevertheless, we may also consider the triple $(T,
\al{},\gst{}{-})$, including the (uniquely existing) operation 
$f\mapsto\gst{(Z,e)}{f}$.

The following is a variant of a theorem from
\cite{DBLP:journals/tcs/MatthesU04}, 
formalized in
\cite{DBLP:journals/corr/AhrensM16}. 
The original theorem required the existence of a right
  adjoint for the functor ${\_} \hcomp Z : [\C,\C] \arr [\C,\C]$ for
  every $\Ptd(\C)$-object $(Z,e)$.
  The present variant replaces that hypothesis on right adjoints by 
  suitable assumptions on $\omega$-cocontinuity.
  \begin{theorem}[Construction of a substitution operation on an initial algebra]\label{thm:subst_on_init}
      Let $\C$ be a category with initial object, binary coproducts and products, and 
      colimits of chains. 
      Let $(H,\theta)$ be a signature over base category $\C$. 
      If $H$ is $\omega$-cocontinuous,
  then an initial 
  $(\constfunctor{\Id{\C}} + H)$-algebra can be constructed via Construction~\ref{thm:init_alg},
  and this initial algebra
  is a heterogeneous substitution
  system for $(H,\th{})$. 
  \end{theorem}
  The proof is done by generalized iteration in Mendler-style (in the
  category-theoretic form introduced by
  \cite[Theorem 1]{DBLP:journals/fac/BirdP99}), both for the existence and the
  uniqueness of $\gst{}{f}$. Here, unlike in the previous work \cite{DBLP:journals/tcs/MatthesU04,DBLP:journals/corr/AhrensM16}, the initial
  algebra has to come from $\omega$-cocontinuity of the
  signature functor. The previous condition on existence of the
  right adjoint in the theorem would not allow to apply it to the category $\set$.

\begin{theorem}[Matthes and Uustalu \cite{DBLP:journals/tcs/MatthesU04}, formalized in \cite{DBLP:journals/corr/AhrensM16}](Construction of a monad from a substitution system)
\label{thm:monad_from_hss}
Let $\C$ be a category with binary coproducts and $(H,\theta)$ a
signature with strength over base category $\C$. If $(T,\alpha)$ is an
hss for $(H,\theta)$, then $T$, together with the canonically
associated $\et{}:\Id{C} \arr T$ as unit and
$\gst{}{1_{(T,\et{})}}:T\hcomp T\to T$ as multiplication, form a
monad.
\end{theorem}
Functional programmers normally do not consider monad multiplication when
studying monads but rather the operation called \verb+bind+. It is
well-known that the formulations of monads with unit and
multiplication and those with unit and bind are equivalent. Given
$A,B:\C$ and a \emph{substitution rule} $f:A\to TB$, the effect of a
parallel substitution with $f$, is then $\gst{}{1_{(T,\et{})}}_B\vcomp
Tf:TA\to TB$, which is the bind operation for argument $f$. 
For $\C=\Set$, this just means that, for an argument $t:TA$, each
free variable occurrence of a variable $a:A$ in $t$ is replaced by the term
$f~a:TB$. The monad laws then become
 conditions for substitution, and
they are guaranteed by the theorem.

\subsection{Binding signatures to monads}\label{sec:monad}

We now recall the results presented in the paper and explain how
to combine them in order to obtain a monad from a binding
signature. 

Let $\C$ be a category with binary products and coproducts.
Let $(I,\arity)$ be a binding signature, by
constructions~\ref{constr_H:sem_sig_from_binding}
and~\ref{constr:sem_sig_from_binding} we obtain a signature with
strength $(H,\theta)$. In \UniMath:

\begin{coqdoccode}
\coqdocemptyline
\coqdocnoindent
\coqdockw{Definition} \coqdef{UniMath.SubstitutionSystems.BindingSigToMonad.BindingSigToSignature}{BindingSigToSignature}{\coqdocdefinition{BindingSigToSignature}} (\coqdocvar{TC} : \coqref{UniMath.CategoryTheory.limits.terminal.Terminal}{\coqdocdefinition{Terminal}} \coqdocvariable{C})\coqdoceol
\coqdocindent{1.00em}
(\coqdocvar{sig} : \coqref{UniMath.SubstitutionSystems.BindingSigToMonad.BindingSig}{\coqdocdefinition{BindingSig}}) (\coqdocvar{CC} : \coqref{UniMath.CategoryTheory.limits.coproducts.Coproducts}{\coqdocdefinition{Coproducts}} (\coqref{UniMath.SubstitutionSystems.BindingSigToMonad.BindingSigIndex}{\coqdocdefinition{BindingSigIndex}} \coqdocvariable{sig}) \coqdocvariable{C}) :\coqdoceol
\coqdocindent{1.00em}
\coqref{UniMath.SubstitutionSystems.Signatures.Signature}{\coqdocdefinition{Signature}} \coqdocvariable{C} \coqdocvariable{hsC}.\coqdoceol
\coqdocemptyline
\end{coqdoccode}

Note that we here require that $\C$ has both binary and $I$-indexed
coproducts, we could instead assume that $\C$ has all indexed
coproducts (as in the statement of
Problem~\ref{prob:sem_sig_from_binding}).

Theorem~\ref{lem:sem_sig_cocont} tells us that $H$ is
$\omega$-cocontinuous:

\begin{coqdoccode}
\coqdocemptyline
\coqdocnoindent
\coqdockw{Lemma} \coqdef{UniMath.SubstitutionSystems.BindingSigToMonad.is omega cocont BindingSigToSignature}{is\_omega\_cocont\_BindingSigToSignature}{\coqdoclemma{is\_omega\_cocont\_BindingSigToSignature}}\coqdoceol
\coqdocindent{1.00em}
(\coqdocvar{TC} : \coqref{UniMath.CategoryTheory.limits.terminal.Terminal}{\coqdocdefinition{Terminal}} \coqdocvariable{C}) (\coqdocvar{CLC} : \coqref{UniMath.CategoryTheory.limits.graphs.colimits.Colims of shape}{\coqdocdefinition{Colims\_of\_shape}} \coqref{UniMath.CategoryTheory.CocontFunctors.nat graph}{\coqdocdefinition{nat\_graph}} \coqdocvariable{C})\coqdoceol
\coqdocindent{1.00em}
(\coqdocvar{HF} : \coqref{UniMath.Foundations.Basics.Preamble.:type scope:'xCExA0' x '..' x ',' x}{\coqdocnotation{\ensuremath{\Pi}}} \coqref{UniMath.Foundations.Basics.Preamble.:type scope:'xCExA0' x '..' x ',' x}{\coqdocnotation{(}}\coqdocvar{F} : \coqref{UniMath.SubstitutionSystems.BindingSigToMonad.BindingSigToMonad.::'[C,C]'}{\coqdocnotation{[}}\coqref{UniMath.SubstitutionSystems.BindingSigToMonad.BindingSigToMonad.::'[C,C]'}{\coqdocnotation{C}}\coqref{UniMath.SubstitutionSystems.BindingSigToMonad.BindingSigToMonad.::'[C,C]'}{\coqdocnotation{,}}\coqref{UniMath.SubstitutionSystems.BindingSigToMonad.BindingSigToMonad.::'[C,C]'}{\coqdocnotation{C}}\coqref{UniMath.SubstitutionSystems.BindingSigToMonad.BindingSigToMonad.::'[C,C]'}{\coqdocnotation{]}}\coqref{UniMath.Foundations.Basics.Preamble.:type scope:'xCExA0' x '..' x ',' x}{\coqdocnotation{),}} \coqref{UniMath.CategoryTheory.CocontFunctors.is omega cocont}{\coqdocdefinition{is\_omega\_cocont}} (\coqdocvariable{constprod\_functor1} \coqdocvariable{F}))\coqdoceol
\coqdocindent{1.00em}
(\coqdocvar{sig} : \coqref{UniMath.SubstitutionSystems.BindingSigToMonad.BindingSig}{\coqdocdefinition{BindingSig}})\coqdoceol
\coqdocindent{1.00em}
(\coqdocvar{CC} : \coqref{UniMath.CategoryTheory.limits.coproducts.Coproducts}{\coqdocdefinition{Coproducts}} (\coqref{UniMath.SubstitutionSystems.BindingSigToMonad.BindingSigIndex}{\coqdocdefinition{BindingSigIndex}} \coqdocvariable{sig}) \coqdocvariable{C}) \coqdoceol
\coqdocindent{1.00em}
(\coqdocvar{PC} : \coqref{UniMath.CategoryTheory.limits.products.Products}{\coqdocdefinition{Products}} (\coqref{UniMath.SubstitutionSystems.BindingSigToMonad.BindingSigIndex}{\coqdocdefinition{BindingSigIndex}} \coqdocvariable{sig}) \coqdocvariable{C}) :\coqdoceol
\coqdocindent{1.00em}
\coqref{UniMath.CategoryTheory.CocontFunctors.is omega cocont}{\coqdocdefinition{is\_omega\_cocont}} (\coqref{UniMath.SubstitutionSystems.BindingSigToMonad.BindingSigToSignature}{\coqdocdefinition{BindingSigToSignature}} \coqdocvariable{TC} \coqdocvariable{sig} \coqdocvariable{CC}).\coqdoceol
\coqdocemptyline
\end{coqdoccode}

Here \coqdocvar{constprod\_functor1 F} denotes the functor that sends
$G$ to $F \times G$. Construction~\ref{thm:init_alg} allows us to
construct an initial algebra for $\constfunctor{\Id{\C}}+H$ under
suitable hypotheses on $\C$:

\begin{coqdoccode}
\coqdocemptyline
\coqdocnoindent
\coqdockw{Definition} \coqdef{UniMath.SubstitutionSystems.BindingSigToMonad.SignatureInitialAlgebra}{SignatureInitialAlgebra}{\coqdocdefinition{SignatureInitialAlgebra}}\coqdoceol
\coqdocindent{1.00em}
(\coqdocvar{IC} : \coqref{UniMath.CategoryTheory.limits.initial.Initial}{\coqdocdefinition{Initial}} \coqdocvariable{C}) (\coqdocvar{CLC} : \coqref{UniMath.CategoryTheory.limits.graphs.colimits.Colims of shape}{\coqdocdefinition{Colims\_of\_shape}} \coqref{UniMath.CategoryTheory.CocontFunctors.nat graph}{\coqdocdefinition{nat\_graph}} \coqdocvariable{C})\coqdoceol
\coqdocindent{1.00em}
(\coqdocvar{H} : \coqref{UniMath.SubstitutionSystems.Signatures.Signature}{\coqdocdefinition{Signature}} \coqdocvariable{C} \coqdocvariable{hsC}) (\coqdocvar{Hs} : \coqref{UniMath.CategoryTheory.CocontFunctors.is omega cocont}{\coqdocdefinition{is\_omega\_cocont}} \coqdocvariable{H}) :\coqdoceol
\coqdocindent{1.00em}
\coqref{UniMath.CategoryTheory.limits.initial.Initial}{\coqdocdefinition{Initial}} (\coqdocvariable{FunctorAlg} (\coqdocvariable{Id\_H} \coqdocvariable{H})).\coqdoceol
\coqdocemptyline
\end{coqdoccode}

By Theorem~\ref{thm:subst_on_init} we then obtain an initial
heterogeneous substitution system:

\begin{coqdoccode}
\coqdocemptyline
\coqdocnoindent
\coqdockw{Definition} \coqdef{UniMath.SubstitutionSystems.BindingSigToMonad.InitialHSS}{InitialHSS}{\coqdocdefinition{InitialHSS}}\coqdoceol
\coqdocindent{1.00em}
(\coqdocvar{IC} : \coqref{UniMath.CategoryTheory.limits.initial.Initial}{\coqdocdefinition{Initial}} \coqdocvariable{C}) (\coqdocvar{CLC} : \coqref{UniMath.CategoryTheory.limits.graphs.colimits.Colims of shape}{\coqdocdefinition{Colims\_of\_shape}} \coqref{UniMath.CategoryTheory.CocontFunctors.nat graph}{\coqdocdefinition{nat\_graph}} \coqdocvariable{C})\coqdoceol
\coqdocindent{1.00em}
(\coqdocvar{H} : \coqref{UniMath.SubstitutionSystems.Signatures.Signature}{\coqdocdefinition{Signature}} \coqdocvariable{C} \coqdocvariable{hsC}) (\coqdocvar{Hs} : \coqref{UniMath.CategoryTheory.CocontFunctors.is omega cocont}{\coqdocdefinition{is\_omega\_cocont}} \coqdocvariable{H}) :\coqdoceol
\coqdocindent{1.00em}
\coqref{UniMath.CategoryTheory.limits.initial.Initial}{\coqdocdefinition{Initial}} (\coqdocvariable{HSS} \coqdocvariable{H}).\coqdoceol
\coqdocemptyline
\end{coqdoccode}

Finally we can obtain a monad from a heterogeneous substitution system
by Theorem~\ref{thm:monad_from_hss}:

\begin{coqdoccode}
\coqdocemptyline
\coqdocnoindent
\coqdockw{Definition} \coqdef{UniMath.SubstitutionSystems.BindingSigToMonad.BindingSigToMonad.Monad from hss}{Monad\_from\_hss}{\coqdocvariable{Monad\_from\_hss}} (\coqdocvar{H} : \coqref{UniMath.SubstitutionSystems.Signatures.Signature}{\coqdocdefinition{Signature}} \coqdocvariable{C} \coqdocvariable{hsC}) : \coqdocvariable{HSS} \coqdocvariable{H} \coqref{UniMath.Foundations.Basics.Preamble.:type scope:x 'xE2x86x92' x}{\coqdocnotation{→}} \coqref{UniMath.CategoryTheory.Monads.Monad}{\coqdocdefinition{Monad}} \coqdocvariable{C}.\coqdoceol
\coqdocemptyline
\end{coqdoccode}

Combining all of this gives us the desired map from binding signatures
to monads:

\begin{coqdoccode}
\coqdocemptyline
\coqdocnoindent
\coqdockw{Definition} \coqdef{UniMath.SubstitutionSystems.BindingSigToMonad.BindingSigToMonad}{BindingSigToMonad}{\coqdocdefinition{BindingSigToMonad}}\coqdoceol
\coqdocindent{1.00em}
(\coqdocvar{TC} : \coqref{UniMath.CategoryTheory.limits.terminal.Terminal}{\coqdocdefinition{Terminal}} \coqdocvariable{C}) (\coqdocvar{IC} : \coqref{UniMath.CategoryTheory.limits.initial.Initial}{\coqdocdefinition{Initial}} \coqdocvariable{C}) (\coqdocvar{CLC} : \coqref{UniMath.CategoryTheory.limits.graphs.colimits.Colims of shape}{\coqdocdefinition{Colims\_of\_shape}} \coqref{UniMath.CategoryTheory.CocontFunctors.nat graph}{\coqdocdefinition{nat\_graph}} \coqdocvariable{C})\coqdoceol
\coqdocindent{1.00em}
(\coqdocvar{HF} : \coqref{UniMath.Foundations.Basics.Preamble.:type scope:'xCExA0' x '..' x ',' x}{\coqdocnotation{\ensuremath{\Pi}}} \coqref{UniMath.Foundations.Basics.Preamble.:type scope:'xCExA0' x '..' x ',' x}{\coqdocnotation{(}}\coqdocvar{F} : \coqref{UniMath.SubstitutionSystems.BindingSigToMonad.BindingSigToMonad.::'[C,C]'}{\coqdocnotation{[}}\coqref{UniMath.SubstitutionSystems.BindingSigToMonad.BindingSigToMonad.::'[C,C]'}{\coqdocnotation{C}}\coqref{UniMath.SubstitutionSystems.BindingSigToMonad.BindingSigToMonad.::'[C,C]'}{\coqdocnotation{,}}\coqref{UniMath.SubstitutionSystems.BindingSigToMonad.BindingSigToMonad.::'[C,C]'}{\coqdocnotation{C}}\coqref{UniMath.SubstitutionSystems.BindingSigToMonad.BindingSigToMonad.::'[C,C]'}{\coqdocnotation{]}}\coqref{UniMath.Foundations.Basics.Preamble.:type scope:'xCExA0' x '..' x ',' x}{\coqdocnotation{),}} \coqref{UniMath.CategoryTheory.CocontFunctors.is omega cocont}{\coqdocdefinition{is\_omega\_cocont}} (\coqdocvariable{constprod\_functor1} \coqdocvariable{F}))\coqdoceol
\coqdocindent{1.00em}
(\coqdocvar{sig} : \coqref{UniMath.SubstitutionSystems.BindingSigToMonad.BindingSig}{\coqdocdefinition{BindingSig}})\coqdoceol
\coqdocindent{1.00em}
(\coqdocvar{PC} : \coqref{UniMath.CategoryTheory.limits.products.Products}{\coqdocdefinition{Products}} (\coqref{UniMath.SubstitutionSystems.BindingSigToMonad.BindingSigIndex}{\coqdocdefinition{BindingSigIndex}} \coqdocvariable{sig}) \coqdocvariable{C})\coqdoceol
\coqdocindent{1.00em}
(\coqdocvar{CC} : \coqref{UniMath.CategoryTheory.limits.coproducts.Coproducts}{\coqdocdefinition{Coproducts}} (\coqref{UniMath.SubstitutionSystems.BindingSigToMonad.BindingSigIndex}{\coqdocdefinition{BindingSigIndex}} \coqdocvariable{sig}) \coqdocvariable{C}) :\coqdoceol
\coqdocindent{1.00em}
\coqref{UniMath.CategoryTheory.Monads.Monad}{\coqdocdefinition{Monad}} \coqdocvariable{C}.\coqdoceol
\coqdocemptyline
\end{coqdoccode}

We see that the category $\C$ needs to have binary coproducts and
products, initial and terminal objects, colimits of chains,
$I$-indexed coproducts and products, and the functor
$G \mapsto F \times G$ has to be $\omega$-cocontinuous.  All of the
assumptions on $\C$ are satisfied by $\Set$. In the formalization we
have implemented special functions instantiated with $\Set$ taking
fewer arguments, in particular:

\begin{coqdoccode}
\coqdocemptyline
\coqdocnoindent
\coqdockw{Definition} \coqdef{UniMath.SubstitutionSystems.BindingSigToMonad.BindingSigToMonadHSET}{BindingSigToMonadHSET}{\coqdocdefinition{BindingSigToMonadHSET}} : \coqref{UniMath.SubstitutionSystems.BindingSigToMonad.BindingSig}{\coqdocdefinition{BindingSig}} \coqref{UniMath.Foundations.Basics.Preamble.:type scope:x 'xE2x86x92' x}{\coqdocnotation{→}} \coqref{UniMath.CategoryTheory.Monads.Monad}{\coqdocdefinition{Monad}} \coqref{UniMath.CategoryTheory.category hset.HSET}{\coqdocabbreviation{HSET}}.\coqdoceol
\coqdocemptyline
\end{coqdoccode}

We end by showing how the framework developed in this paper can be
used to conveniently obtain monads from binding signatures for two
well-known languages.

\begin{example}[Untyped lambda calculus]

As explained in the beginning of the paper
the binding
signature for the untyped lambda calculus is given by
$I:= \{ \app{}, \abs{} \}$ and the arity function
\[ \app{} \mapsto [0,0] \enspace , \enspace \abs{} \mapsto [1] \enspace .\]

In \UniMath we implement this as a \verb+bool+-indexed family:

\begin{coqdoccode}
\coqdocemptyline
\coqdocnoindent
\coqdockw{Definition} \coqdef{UniMath.SubstitutionSystems.LamFromBindingSig.LamSig}{LamSig}{\coqdocdefinition{LamSig}} : \coqref{UniMath.SubstitutionSystems.BindingSigToMonad.BindingSig}{\coqdocdefinition{BindingSig}} :=\coqdoceol
\coqdocindent{1.00em}
\coqref{UniMath.SubstitutionSystems.BindingSigToMonad.mkBindingSig}{\coqdocdefinition{mkBindingSig}} \coqref{UniMath.Foundations.Basics.PartB.isdeceqbool}{\coqdoclemma{isdeceqbool}} (\coqdockw{fun} \coqdocvar{b} \ensuremath{\Rightarrow} \coqdockw{if} \coqdocvariable{b} \coqdockw{then} 0 \coqref{UniMath.SubstitutionSystems.LamFromBindingSig.Lam.::x '::' x}{\coqdocnotation{::}} 0 \coqref{UniMath.SubstitutionSystems.LamFromBindingSig.Lam.::x '::' x}{\coqdocnotation{::}} \coqref{UniMath.SubstitutionSystems.LamFromBindingSig.Lam.::'[]'}{\coqdocnotation{[]}} \coqdockw{else} 1 \coqref{UniMath.SubstitutionSystems.LamFromBindingSig.Lam.::x '::' x}{\coqdocnotation{::}} \coqref{UniMath.SubstitutionSystems.LamFromBindingSig.Lam.::'[]'}{\coqdocnotation{[]}}).\coqdoceol
\coqdocemptyline
\end{coqdoccode}

From this we obtain a signature with strength:

\begin{coqdoccode}
\coqdocemptyline
\coqdocnoindent
\coqdockw{Definition} \coqdef{UniMath.SubstitutionSystems.LamFromBindingSig.LamSignature}{LamSignature}{\coqdocdefinition{LamSignature}} : \coqref{UniMath.SubstitutionSystems.Signatures.Signature}{\coqdocdefinition{Signature}} \coqref{UniMath.CategoryTheory.category hset.HSET}{\coqdocabbreviation{HSET}} \coqref{UniMath.CategoryTheory.category hset.has homsets HSET}{\coqdoclemma{has\_homsets\_HSET}} :=\coqdoceol
\coqdocindent{1.00em}
\coqref{UniMath.SubstitutionSystems.BindingSigToMonad.BindingSigToSignatureHSET}{\coqdocdefinition{BindingSigToSignatureHSET}} \coqref{UniMath.SubstitutionSystems.LamFromBindingSig.LamSig}{\coqdocdefinition{LamSig}}.\coqdoceol
\coqdocemptyline
\end{coqdoccode}

Using this we can add variables in order to get a representation of
the complete syntax of the untyped lambda calculus. We also get an
initial algebra from this functor by Construction~\ref{thm:init_alg}:

\begin{coqdoccode}
\coqdocemptyline
\coqdocnoindent
\coqdockw{Definition} \coqdef{UniMath.SubstitutionSystems.LamFromBindingSig.LamFunctor}{LamFunctor}{\coqdocdefinition{LamFunctor}} : \coqref{UniMath.CategoryTheory.functor categories.functor}{\coqdocdefinition{functor}} \coqref{UniMath.SubstitutionSystems.LamFromBindingSig.Lam.::'HSET2'}{\coqdocnotation{HSET2}} \coqref{UniMath.SubstitutionSystems.LamFromBindingSig.Lam.::'HSET2'}{\coqdocnotation{HSET2}} := \coqdocvariable{Id\_H} \coqref{UniMath.SubstitutionSystems.LamFromBindingSig.LamSignature}{\coqdocdefinition{LamSignature}}.\coqdoceol
\coqdocemptyline
\coqdocnoindent
\coqdockw{Lemma} \coqdef{UniMath.SubstitutionSystems.LamFromBindingSig.lambdaFunctor Initial}{lambdaFunctor\_Initial}{\coqdoclemma{lambdaFunctor\_Initial}} : \coqref{UniMath.CategoryTheory.limits.initial.Initial}{\coqdocdefinition{Initial}} (\coqref{UniMath.CategoryTheory.FunctorAlgebras.FunctorAlg}{\coqdocabbreviation{FunctorAlg}} \coqref{UniMath.SubstitutionSystems.LamFromBindingSig.LamFunctor}{\coqdocdefinition{LamFunctor}}).\coqdoceol
\coqdocemptyline
\end{coqdoccode}

Here \coqdocnotation{HSET2} is notation for $[\Set,\Set]$. Using this we can
define constructors and propositional computation rules as for
lists. 
We omit these due to space constraints but the interested
reader can consult the formalization.
Finally we also get a substitution monad:

\begin{coqdoccode}
\coqdocemptyline
\coqdocnoindent
\coqdockw{Definition} \coqdef{UniMath.SubstitutionSystems.LamFromBindingSig.LamMonad}{LamMonad}{\coqdocdefinition{LamMonad}} : \coqref{UniMath.CategoryTheory.Monads.Monad}{\coqdocdefinition{Monad}} \coqref{UniMath.CategoryTheory.category hset.HSET}{\coqdocabbreviation{HSET}} := \coqref{UniMath.SubstitutionSystems.BindingSigToMonad.BindingSigToMonadHSET}{\coqdocdefinition{BindingSigToMonadHSET}} \coqref{UniMath.SubstitutionSystems.LamFromBindingSig.LamSig}{\coqdocdefinition{LamSig}}.\coqdoceol
\coqdocemptyline
\end{coqdoccode}

\end{example}

\begin{example}[Raw syntax of Martin-Löf type theory]\label{ex:mltt}

We have also implemented a more substantial example: the raw syntax of
Martin-Löf type theory as presented in~\cite{MLTT79}. This syntax has
$\Pi$-types, $\Sigma$-types, coproduct types, identity types, finite
types, natural numbers, W-types and an infinite hierarchy of
universes. See Table~\ref{tab:mltt} for a summary of this language.
\begin{table}[h!]
\caption{This is the syntax as presented on page 158 of~\cite{MLTT79}.}
\label{tab:mltt}
\begin{tabular}{lll}
\hline\noalign{\smallskip}
Types & Concrete syntax & Binding arities  \\
\noalign{\smallskip}\hline\noalign{\smallskip}
Pi types & (Πx:A)B, (λx)b, (c)a & [0,1], [1], [0,0]\\
Sigma types & (Σx:A)B, (a,b), (Ex,y)(c,d)  & [0,1], [0,0], [0,2]\\
Sum types & A + B, i(a), j(b), (Dx,y)(c,d,e) & [0,0], [0], [0], [0,1,1]\\
Id types & I(A,a,b), r, J(c,d) & [0,0,0], [], [0,0]\\
Fin types & $N_i$, $0_i \cdots (i-1)_i$, $R_i$(c,$c_0$,...,$c_{i-1}$) & [], [] $\cdots$ [], [0,0,...,0]\\
Natural numbers & N, 0, a', (Rx,y)(c,d,e) & [], [], [0], [0,0,2]\\
W-types & (Wx∈A)B, sup(a,b), (Tx,y,z)(c,d) & [0,1], [0,0], [0,3]\\
Universes & $U_0$, $U_1$, ... & [], [], ...\\
\noalign{\smallskip}\hline
\end{tabular}
\end{table}

Because there are both infinitely many finite types and universes the
syntax has infinitely many constructors. This is the reason why we
above consider families of lists of natural numbers and 
indexed coproducts. Note that all of the operations take finitely many
arguments which is why we don't need to also consider infinite arities
and indexed products.

We define the binding signatures for each of these types
separately. Below is the code for $\Pi$- and $\Sigma$-types:

\begin{coqdoccode}
\coqdocemptyline
\coqdocnoindent
\coqdockw{Definition} \coqdef{UniMath.SubstitutionSystems.MLTT79.PiSig}{PiSig}{\coqdocdefinition{PiSig}} : \coqref{UniMath.SubstitutionSystems.BindingSigToMonad.BindingSig}{\coqdocdefinition{BindingSig}} :=\coqdoceol
\coqdocindent{1.00em}
\coqref{UniMath.SubstitutionSystems.BindingSigToMonad.mkBindingSig}{\coqdocdefinition{mkBindingSig}} (\coqref{UniMath.Foundations.Combinatorics.StandardFiniteSets.isdeceqstn}{\coqdoclemma{isdeceqstn}} 3) (\coqref{UniMath.Foundations.Combinatorics.StandardFiniteSets.three rec}{\coqdocdefinition{three\_rec}} \coqref{UniMath.SubstitutionSystems.MLTT79.MLTT79.::'[0,1]'}{\coqdocnotation{[0,1]}} \coqref{UniMath.SubstitutionSystems.MLTT79.MLTT79.::'[1]'}{\coqdocnotation{[1]}} \coqref{UniMath.SubstitutionSystems.MLTT79.MLTT79.::'[0,0]'}{\coqdocnotation{[0,0]}}).\coqdoceol
\coqdocemptyline
\coqdocnoindent
\coqdockw{Definition} \coqdef{UniMath.SubstitutionSystems.MLTT79.SigmaSig}{SigmaSig}{\coqdocdefinition{SigmaSig}} : \coqref{UniMath.SubstitutionSystems.BindingSigToMonad.BindingSig}{\coqdocdefinition{BindingSig}} :=\coqdoceol
\coqdocindent{1.00em}
\coqref{UniMath.SubstitutionSystems.BindingSigToMonad.mkBindingSig}{\coqdocdefinition{mkBindingSig}} (\coqref{UniMath.Foundations.Combinatorics.StandardFiniteSets.isdeceqstn}{\coqdoclemma{isdeceqstn}} 3) (\coqref{UniMath.Foundations.Combinatorics.StandardFiniteSets.three rec}{\coqdocdefinition{three\_rec}} \coqref{UniMath.SubstitutionSystems.MLTT79.MLTT79.::'[0,1]'}{\coqdocnotation{[0,1]}} \coqref{UniMath.SubstitutionSystems.MLTT79.MLTT79.::'[0,0]'}{\coqdocnotation{[0,0]}} \coqref{UniMath.SubstitutionSystems.MLTT79.MLTT79.::'[0,2]'}{\coqdocnotation{[0,2]}}).\coqdoceol
\coqdocemptyline
\end{coqdoccode}

Here the function \coqdocdefinition{three\_rec} \verb+a b c+ performs
case analysis and returns one of \verb+a+, \verb+b+ or \verb+c+. We
then combine all of these binding signatures by taking their sum:

\begin{coqdoccode}
\coqdocemptyline
\coqdocnoindent
\coqdockw{Definition} \coqdef{UniMath.SubstitutionSystems.MLTT79.MLTT79Sig}{MLTT79Sig}{\coqdocdefinition{MLTT79Sig}} := \coqref{UniMath.SubstitutionSystems.MLTT79.PiSig}{\coqdocdefinition{PiSig}} \coqref{UniMath.SubstitutionSystems.MLTT79.::x '++' x}{\coqdocnotation{++}} \coqref{UniMath.SubstitutionSystems.MLTT79.SigmaSig}{\coqdocdefinition{SigmaSig}} \coqref{UniMath.SubstitutionSystems.MLTT79.::x '++' x}{\coqdocnotation{++}} \coqref{UniMath.SubstitutionSystems.MLTT79.SumSig}{\coqdocdefinition{SumSig}} \coqref{UniMath.SubstitutionSystems.MLTT79.::x '++' x}{\coqdocnotation{++}} \coqref{UniMath.SubstitutionSystems.MLTT79.IdSig}{\coqdocdefinition{IdSig}} \coqref{UniMath.SubstitutionSystems.MLTT79.::x '++' x}{\coqdocnotation{++}}\coqdoceol
\coqdocindent{12.50em}
\coqref{UniMath.SubstitutionSystems.MLTT79.FinSig}{\coqdocdefinition{FinSig}} \coqref{UniMath.SubstitutionSystems.MLTT79.::x '++' x}{\coqdocnotation{++}} \coqref{UniMath.SubstitutionSystems.MLTT79.NatSig}{\coqdocdefinition{NatSig}} \coqref{UniMath.SubstitutionSystems.MLTT79.::x '++' x}{\coqdocnotation{++}} \coqref{UniMath.SubstitutionSystems.MLTT79.WSig}{\coqdocdefinition{WSig}} \coqref{UniMath.SubstitutionSystems.MLTT79.::x '++' x}{\coqdocnotation{++}} \coqref{UniMath.SubstitutionSystems.MLTT79.USig}{\coqdocdefinition{USig}}.\coqdoceol
\coqdocemptyline
\end{coqdoccode}

Finally we also obtain a substitution monad on $\Set$ for this
language:

\begin{coqdoccode}
\coqdocemptyline
\coqdocnoindent
\coqdockw{Definition} \coqdef{UniMath.SubstitutionSystems.MLTT79.MLTT79Monad}{MLTT79Monad}{\coqdocdefinition{MLTT79Monad}} : \coqref{UniMath.CategoryTheory.Monads.Monad}{\coqdocdefinition{Monad}} \coqref{UniMath.CategoryTheory.category hset.HSET}{\coqdocabbreviation{HSET}} := \coqref{UniMath.SubstitutionSystems.BindingSigToMonad.BindingSigToMonadHSET}{\coqdocdefinition{BindingSigToMonadHSET}} \coqref{UniMath.SubstitutionSystems.MLTT79.MLTT79Sig}{\coqdocdefinition{MLTT79Sig}}.\coqdoceol
\coqdocemptyline
\end{coqdoccode}

\end{example}

\section{Conclusion and future work}

\subsection{Conclusions}\label{sec:conclusion}

  We have formalized some classical category-theoretic results on 
  the construction of initial algebras, as well as on cocontinuity of functors.
Maybe surprisingly, the formalization of results yielding $\omega$-cocontinuous functors as input
to the construction of initial algebras proved to be much more difficult than the construction of colimits in $\set$ as output of that theorem.

Our formalization has been integrated into the \UniMath
library. Statistics related to the contributions of this paper have
been summarized in Table~\ref{table:loc}.\footnote{The script for computing these statistics can be found at:\\
  \url{https://github.com/mortberg/UniMath/tree/locscript/loc}} 
The first three columns show lines of code and the last two show the
number of vernacular commands.
\begin{table}[h!]
\caption{Statistics for the formalization.}
\label{table:loc}
\begin{tabular}{lllll}
\hline\noalign{\smallskip}
Specification & Proof & Comments & \verb+Definition+ & \verb+Lemma+ and \verb+Theorem+  \\
\noalign{\smallskip}\hline\noalign{\smallskip}
3623 & 5283 & 1538 & 649 & 482 \\
\noalign{\smallskip}\hline
\end{tabular}
\end{table}

  Our datatypes come with a recursion principle, 
  defined via the universal property of the datatype as an initial 
  algebra. 
  This recursion principle allows us to define maps such as \texttt{foldr}
  for lists. Those maps satisfy the usual computation rules \emph{judgmentally},
  provided that 
  \begin{enumerate}
   \item the output type is one of the predefined types of \UniMath; and
   \item the computation is done \emph{lazily}.
  \end{enumerate}
  An instance of this is the length function for lists, the output type of which is
  the type $\texttt{nat}$ of natural numbers, defined as an inductive \Coq type.
  Maps whose output type is a datatype constructed via our framework do
  not compute to a normal form.
  An example of such a map is the function concatenating two lists into one list.
  Trying to compute the normal form of such a concatenated list leads
  to memory exhaustion.
  However, we can still reason about such maps by rewriting, that is,
  by replacing computational steps by a suitable lemma
  stating this step as a propositional equality.
  This is precisely how many recursive maps are handled in \textsc{SSReflect}~\cite{ssr}.
  There, computation of recursive maps is deliberately blocked for efficiency reasons in order to avoid
  too much unfolding. Instead, computation steps are simulated by 
  applying suitable rewriting lemmas.
  This indicates that the lack of a computable normal form
  for the inhabitants of our datatypes is not an obstacle
  for mathematical reasoning about the maps that we define on those 
  datatypes. 

    In the proofs and constructions presented here, the univalence principle is only used in 
    a restricted form: 
    \begin{itemize}
     \item \emph{function extensionality}, a consequence of univalence, is used in many places;
     \item the construction of set-level quotients by Voevodsky makes use of the \emph{univalence
          principle for propositions}: two propositions are equal when they are logically equivalent.
          Consequently, our construction of colimits in the category of sets 
          also depends on the univalence axiom for propositions.
    \end{itemize}
 An alternative to the use of these axioms (by admitting the univalence axiom) would have been to work 
 with \emph{setoids}. There, the idea is to abandon the identity type; instead, each type
 comes equipped with its own equivalence relation, reflecting the intended ``equality''.
 This would have been extremely cumbersome, since in that case, one needs to postulate respectively prove
 that any operation respects the equivalence in the source and target. 
 For the identity type, on the other hand, this respectfulness is automatic.

Another alternative would be to work in a system where these are
provable, and hence not axioms anymore, like Cubical Type
Theory~\cite{cubicaltt}. The additional judgmental equalities in such
a system could potentially simplify some proofs, but that needs
to be studied further.

\subsection{Future work}

In this section, we lay out some plans for future work and connections to other work.

  \subsubsection{Initiality for the constructed monad}  
  As illustrated in Section~\ref{sec:monad}, we have formalized a mechanism that,
  when provided with a binding signature, yields the associated \enquote{term monad}
  and a suitable recursion principle for defining maps from the term monad to 
  other (families of) sets. 
  This recursion principle stems from the universal property of initiality
  that the functor underlying the monad satisfies.
  However, the constructed monad itself has not, in the present work, been equipped with a universal property.
  
  Hirschowitz and Maggesi \cite{DBLP:conf/wollic/HirschowitzM07,DBLP:journals/iandc/HirschowitzM10}
  equip the term monad
  of a signature $S$ with a universal property by considering a 
    \emph{category of representations} of a given signature.
  A representation of $S$ is given by any monad $T$ and a family of module morphisms of suitable type over $T$.
  We should be able to formalize Hirschowitz and Maggesi's initiality theorem using the monad
  we have constructed in the present work.
  
  \subsubsection{Generalization to multi-sorted binding signatures}\label{sec:multi-sorted}
  
  The notion of binding signature considered in this paper does not incorporate a notion of \emph{typing}.
  Suitable generalizations to typed (or multi-sorted) signatures have been considered, for instance, in \cite{ahrens_zsido}. 
  In general, a multi-sorted signature contains not only information about the number of bound variables,
  but also of their types. Furthermore, it specifies an output type for each constructor.
  Multi-sorted binding signatures allow to specify languages such as the simply-typed lambda calculus and PCF (Dana Scott's language for ``Programming Computable Functions'').
We are currently working on extending our notions of signature, as well as the construction of 
initial algebras, to the multi-sorted setting.

  \subsubsection{Connection to Voevodsky's $C$-systems}\label{sec:c-systems}
  Voevodsky is currently considering Cartmell's \emph{contextual categories} \cite{DBLP:journals/apal/Cartmell86}, 
  under the name of \enquote{$C$-system},
  for a mathematical description of type theories (see, e.g., \cite{c-sys-universe,c-sys-from-module}).
  In particular, one of Voevodsky's goals is to give a precise construction of the $C$-system formed by the syntax of a given type theory.
  One step of this construction is given in \cite{c-sys-from-module}, where he
  constructs a $C$-system from a pair of a monad on $\set$ and a module over that monad with values in $\set$. 
  Such a pair can be constructed from a monad on $\Set^2$ and a choice of a set.

  It is our goal to formalize this construction in \UniMath, and to apply it to the term monads
  of 2-sorted signatures obtained via the generalization envisioned in Section~\ref{sec:multi-sorted}.
  We will thus obtain, for any suitable 2-sorted signature, a $C$-system of raw syntax of that signature.

\subsubsection*{Acknowledgements:}
 We thank Dan Grayson and Vladimir Voevodsky for helpful discussion on the subject matter. 
 We particularly thank Paige North for pointing to a size problem in an earlier version of one of our categorical 
 constructions during the writing phase of this article.
 Much of the formalization work on this article was done at the Institute for Advanced Study.
 We would like to thank the IAS for providing a pleasant and productive work environment.

\bibliographystyle{plain}

\end{document}